\documentclass[journal,10pt]{IEEEtran}

\usepackage{cite}
\usepackage{graphicx,color,epsfig,rotating}
\usepackage{amsfonts,amsmath,amssymb,bbm}
\usepackage{setspace}
\usepackage{algorithm}
\usepackage[algo2e]{algorithm2e} 
\usepackage{subcaption}
\usepackage{mdwtab}
\usepackage{placeins}
\usepackage{psfrag, graphicx}
\usepackage[latin1]{inputenc}
\usepackage{multirow}
\usepackage{stfloats}
\usepackage{tabularx} 
\usepackage{booktabs} 
\usepackage{url}
\usepackage{bm}
\usepackage{bbding}
\usepackage{pifont}
\usepackage{wasysym}
\usepackage{pifont}
\usepackage{lipsum}
\usepackage{wasysym}

\long\def\comment#1{}

\newtheorem{example}{Example}
\newtheorem{theorem}{Theorem}
\newtheorem{lemma}{Lemma}
\newtheorem{claim}{Claim}

\newtheorem{remark}{Remark}

\makeatletter
\newcommand{\subalign}[1]{%
  \vcenter{%
    \Let@ \restore@math@cr \default@tag
    \baselineskip\fontdimen10 \scriptfont\tw@
    \advance\baselineskip\fontdimen12 \scriptfont\tw@
    \lineskip\thr@@\fontdimen8 \scriptfont\thr@@
    \lineskiplimit\lineskip
    \ialign{\hfil$\m@th\scriptstyle##$&$\m@th\scriptstyle{}##$\crcr
      #1\crcr
    }%
  }
}
\newcommand{\thickhline}{%
    \noalign {\ifnum 0=`}\fi \hrule height 1pt
    \futurelet \reserved@a \@xhline
}
\newcolumntype{"}{@{\hskip\tabcolsep\vrule width 1pt\hskip\tabcolsep}}

\makeatother

\newfont{\bbb}{msbm10 scaled 700}

\newfont{\bb}{msbm10 scaled 1100}

\let \bksl\backslash

\newcommand{\ellg}[1]{\ell^{(#1)}}
\newcommand{\gammag}[1]{\gamma_{#1}}
\newcommand{\alphabmg}[1]{\bm{\alpha}^{(#1)}}

\newcommand{\Gcg}{\Gc_g}
\newcommand{\FmT}{\Fm^{\rm T}}

\newcommand{\eqvlt}{equivalent\xspace}

\newcommand{\orig}{original\xspace}

\newcommand{\compos}{components\xspace}

\newcommand{\randvar}{random variable\xspace}
\newcommand{\randvars}{random variables\xspace}

\newcommand{\User}[1]{User~$#1$}
\newcommand{\user}[1]{user~$#1$}

\usepackage{xcolor} 

\newcommand{\ingen}{in general\xspace}

\newcommand{\Insum}{In summary\xspace}

\newcommand{\folwg}{following\xspace}

\newcommand{\App}{Appendix\xspace}

\newcommand{\sbst}{subset\xspace}

\newcommand{\sbsts}{subsets\xspace}

\newcommand{\decrs}{decrease\xspace}
\newcommand{\decrss}{decreases\xspace}

\newcommand{\decrsg}{decreasing\xspace}
\newcommand{\incrs}{increase\xspace}

\newcommand{\incrss}{increases\xspace}

\newcommand{\ppsd}{proposed\xspace}

\newcommand{\stochaly}{stochastically\xspace}
\newcommand{\stodom}{stochastic dominance\xspace}
\newcommand{\hygeo}{hypergeometric\xspace}

\newcommand{\selec}{selection\xspace}

\newcommand{\selecs}{selections\xspace}

\newcommand{\intmdt}{intermediate\xspace}

\newcommand{\relato}{relative to\xspace}
\newcommand{\Sec}{Section\xspace}

\newcommand{\Movr}{Moreover\xspace}

\newcommand{\deli}{delivery\xspace}
\newcommand{\Deli}{Delivery\xspace}
\newcommand{\suth}{such that\xspace}
\newcommand{\soth}{so that\xspace}

\newcommand{\uneq}{unequal\xspace}

\newcommand{\ropt}{rate-optimal\xspace}

\newcommand{\satisfs}{satisfies\xspace}

\newcommand{\possb}{possible\xspace}

\newcommand{\sbtp}{subtype\xspace}
\newcommand{\sbtps}{subtypes\xspace}

\newcommand{\cpgrp}{coupled group\xspace}
\newcommand{\cpgrps}{coupled groups\xspace}

\newcommand{\memo}{memory\xspace}

\newcommand{\mgrpts}{\mgrp types\xspace}
\newcommand{\sbft}{subfile type\xspace}

\newcommand{\sbfts}{subfile types\xspace}

\newcommand{\accrdto}{according to\xspace}
\newcommand{\Accrdto}{According to\xspace}

\newcommand{\Rmk}{Remark\xspace}

\newcommand{\uset}{unique set\xspace}
\newcommand{\usets}{unique sets\xspace}
\newcommand{\fsf}{further splitting factor\xspace}

\newcommand{\sect}{section\xspace}

\newcommand{\Fex}{For example\xspace}

\newcommand{\ex}{example\xspace}

\newcommand{\strtgs}{strategies\xspace}

\newcommand{\ptf}{PT framework\xspace}

\newcommand{\illutn}{illustration\xspace}

\newcommand{\PTd}{PT design\xspace}

\newcommand{\PTD}{PT Design\xspace}

\newcommand{\achvs}{achieves\xspace}
\newcommand{\achvd}{achieved\xspace}
\newcommand{\achv}{achieve\xspace}

\newcommand{\achvb}{achievable\xspace}

\newcommand{\Itf}{In the following\xspace}

\newcommand{\seq}{sequence\xspace}
\newcommand{\seqs}{sequences\xspace}

\newcommand{\indeptly}{independently\xspace}

\newcommand{\grpg}{grouping\xspace}
\newcommand{\grpgs}{groupings\xspace}

\newcommand{\schms}{schemes\xspace}
\newcommand{\schm}{scheme\xspace}

\newcommand{\Dtm}{Determine\xspace}

\newcommand{\dof}{DoF\xspace}

\newcommand{\corrspdg}{corresponding\xspace}

\newcommand{\bcuz}{because\xspace}

\newcommand{\Aar}{As a result\xspace}
\newcommand{\sbfs}{subfiles\xspace}
\newcommand{\sbf}{subfile\xspace}
\newcommand{\Sbf}{Subfile\xspace}

\newcommand{\perms}{permutations\xspace}
\newcommand{\grps}{groups\xspace}
\newcommand{\grp}{group\xspace}

\newcommand{\frmwk}{framework\xspace}
\newcommand{\Frmwk}{Framework\xspace}
\newcommand{\Pkt}{Packet\xspace}
\newcommand{\Pkts}{Packets\xspace}
\newcommand{\pkt}{packet\xspace}
\newcommand{\pkts}{packets\xspace}

\newcommand{\req}{requirement\xspace}

\newcommand{\ie}{i.e.\xspace}
\newcommand{\agg}{aggregate\xspace}
\newcommand{\mtcst}{multicast\xspace}

\newcommand{\CoCa}{Coded Caching\xspace}

\newcommand{\coca}{coded caching\xspace}

\newcommand{\af}{as follows\xspace}

\newcommand{\Wlog}{Without loss of generality\xspace}
\newcommand{\Msp}{More specifically\xspace}

\newcommand{\Specly}{Specifically\xspace}

\newcommand{\resp}{respectively\xspace}

\newcommand{\diffce}{difference\xspace}

\newcommand{\diff}{different\xspace}

\newcommand{\jsch}{JCM scheme\xspace}
\newcommand{\msch}{MAN scheme\xspace}
\newcommand{\jcm}{JCM\xspace}
\newcommand{\man}{MAN\xspace}

\newcommand{\Inadd}{In addition\xspace}

\newcommand{\Ip}{In particular\xspace}

\newcommand{\tx}{transmitter\xspace}
\newcommand{\txs}{transmitters\xspace}

\newcommand{\txn}{transmission\xspace}

\newcommand{\Tx}{Transmitter\xspace}

\newcommand{\ssgain}{subfile saving gain\xspace}
\newcommand{\fssgain}{further splitting saving gain\xspace}

\newcommand{\mgrp}{multicast group\xspace}
\newcommand{\mgrps}{multicast groups\xspace}

\newcommand{\sym}{symmetric\xspace}

\newcommand{\asympt}{asymptotic\xspace}

\newcommand{\Asympt}{Asymptotic\xspace}

\newcommand{\cmsg}{coded message\xspace}
\newcommand{\cmsgs}{coded messages\xspace}
\newcommand{\Cmsg}{Coded message\xspace}

\newcommand{\comm}{communication\xspace}

\newcommand{\etal}{\emph{et al.}\xspace}

\newcommand{\FinLen}{Finite-Length\xspace}
\newcommand{\Sbp}{Subpacketization\xspace}
\newcommand{\sbp}{subpacketization\xspace}

\newcommand{\Homo}{Homogeneous\xspace}
\newcommand{\homo}{homogeneous\xspace}
\newcommand{\Het}{Heterogeneous\xspace}
\newcommand{\het}{heterogeneous\xspace}

\newcommand{\sota}{state-of-the-art\xspace}

\newcommand{\Thm}{Theorem\xspace}

\newcommand{\red}{reduction\xspace}

\newcommand{\params}{parameters\xspace}

\newcommand{\Thf}{Therefore\xspace}

\newcommand{\fs}{further splitting\xspace}

\newcommand{\alphaglobal}{\mbox{$\bm{\alpha}^{\rm  global}$}}

\newcommand{\alphabm}{\mbox{$\bm{\alpha}$}}
\newcommand{\alphahat}{\mbox{$\widehat{\alpha}$}}
\newcommand{\ith}{\mbox{$i^{\rm th}$}}
\newcommand{\gth}{\mbox{$g^{\rm th}$}}
\newcommand{\jth}{\mbox{$j^{\rm th}$}}
\newcommand{\kth}{\mbox{$k^{\rm th}$}}

\newcommand{\Fpt}{\mbox{$F_{\rm PT}$}}
\newcommand{\Fjcm}{\mbox{$F_{\rm JCM}$}}

\newcommand{\kmax}{\mbox{$k_{\rm max}$}}

\newcommand{\Fman}{\mbox{$F_{\rm MAN}$}}
\newcommand{\Rjcm}{\mbox{$R_{\rm JCM}$}}
\newcommand{\Rman}{\mbox{$R_{\rm MAN}$}}
\newcommand{\nd}{\mbox{$N_{\rm d}$}} 
\newcommand{\hatu}{\mbox{$\widehat{\Uc}$}}
\newcommand{\Utx}{\mbox{$\Uc_{\rm Tx}$}}
\newcommand{\Dtx}{\mbox{$\Dc_{\rm Tx}$}}

\newcommand{\av}{{\bf a}}
\newcommand{\bv}{{\bf b}}

\newcommand{\dv}{{\bf d}}

\newcommand{\qv}{{\bf q}}

\newcommand{\sv}{{\bf s}}

\newcommand{\vv}{{\bf v}}

\newcommand{\Am}{{\bf A}}
\newcommand{\Bm}{{\bf B}}

\newcommand{\Fm}{{\bf F}}

\newcommand{\Ac}{{\cal A}}
\newcommand{\Bc}{{\cal B}}

\newcommand{\Dc}{{\cal D}}

\newcommand{\Gc}{{\cal G}}

\newcommand{\Ic}{{\cal I}}

\newcommand{\Pc}{{\cal P}}
\newcommand{\Qc}{{\cal Q}}

\newcommand{\Sc}{{\cal S}}
\newcommand{\Tc}{{\cal T}}
\newcommand{\Uc}{{\cal U}}

\newcommand{\Vc}{{\cal V}}

\newcommand{\constrtn}{construction\xspace}

\newcommand{\tp}{type\xspace}
\newcommand{\tps}{types\xspace}

\newcommand{\dd}{D2D\xspace}

\newcommand{\invos}{involves\xspace}
\newcommand{\invod}{involved\xspace}

\newcommand{\ists}{involved subfile type set\xspace}

\newcommand{\MemoConst}{Memory Constraint\xspace}

\newcommand{\memoconst}{memory constraint\xspace}
\newcommand{\Memoconst}{Memory constraint\xspace}

\newcommand{\gfsf}{global FS factor\xspace}
\newcommand{\gfsv}{global FS vector\xspace}

\newcommand{\gfsvs}{global FS vectors\xspace}

\newcommand{\lfsfs}{local FS factors\xspace}

\newcommand{\lfsvs}{local FS vectors\xspace}

\newcommand{\vlcm}{vector least common multiple\xspace}

\newcommand{\simully}{simultaneously\xspace}

\newcommand{\Deltam}{\hbox{\boldmath$\Delta$}}

\newcommand{\eqdef}{\stackrel{\Delta}{=}}

\newcommand{\be}{\begin{equation}}
\newcommand{\ee}{\end{equation}}
\newcommand{\bea}{\begin{eqnarray}}
\newcommand{\eea}{\end{eqnarray}}

\def \type{\ttt{type}}

\let \rig\right
\let\lef \left 

\let\ttt\texttt
\let\trm\textrm
\let\tbf\textbf
\let\tit\textit
\let\mc\mathcal

\let\mbb\mathbb

\let \mrm \mathrm

\SetKwInput{KwData}{Input}
\SetKwInput{KwResult}{Output}
\allowdisplaybreaks 
\graphicspath{{./images/}}

\begin{document}

\title{Reducing Subpacketization in  Device-to-Device Coded Caching via Heterogeneous File Splitting}

\author{
Xiang Zhang,~\IEEEmembership{Member,~IEEE},
Giuseppe Caire,~\IEEEmembership{Fellow,~IEEE},
and Mingyue Ji,~\IEEEmembership{Member,~IEEE}

\thanks{Part of this work was presented in the conference paper~\cite{9174215}.}
\thanks{X. Zhang and G. Caire are with the Department of Electrical Engineering and Computer Science, Technical University of Berlin, 10623 Berlin, Germany (e-mail: \{xiang.zhang, caire\}@tu-berlin.de).
M. Ji is with the Department of Electrical and Computer Engineering, University of Florida, Gainesville, FL 32611, USA
(e-mail: mingyueji@ufl.edu).
This work was done while X. Zhang and M. Ji were with the Department of Electrical and Computer Engineering, University of Utah, Salt Lake City, UT 84112, USA. \tit{(Corresponding author: Xiang Zhang.)}
}
}

\maketitle
\IEEEpeerreviewmaketitle

\begin{abstract}
The packet type (PT)-based framework~\cite{zhang2026taming} provides a systematic and principled approach to designing device-to-device (D2D) coded caching schemes that  achieve reduced \sbp while preserving the optimal communication rate. 
However, existing PT designs rely exclusively on homogeneous \sbp, where all packets have an identical size regardless of their types. This restriction limits the achievable \sbp reduction in certain parameter regimes.
In this paper, we extend the PT framework to \emph{heterogeneous} \sbp, allowing packet sizes to vary across types under a refined type classification. The packet sizes, in conjunction with user grouping and multicast transmitter selection, are jointly optimized to minimize the overall \sbp level while preserving the optimal rate. 
Based on the heterogeneous PT framework, we construct a new class of D2D coded caching schemes 
for $(K, KM/N)=(2q+1, 2r)$ with $q,r \in  \mathbb{N}_+$, where $K,N$ and $M$ denote the number of users, files and cache memory size, respectively.
The proposed construction  achieves 
a constant-factor reduction in \sbp compared to  the Ji-Caire-Molisch (JCM) caching scheme~\cite{ji2016fundamental} and complements existing PT designs that are not applicable in this parameter regime.
\end{abstract}


\section{Introduction}
\label{section:intro}


Device-to-device (D2D) coded caching~\cite{ji2016fundamental} extends the shared-link model of Maddah-Ali and Niesen~\cite{maddah2014fundamental} (referred to as the \tit{MAN \schm}) to serverless networks, where users take turns transmitting to satisfy each other's demands.
Let $K,N$ and $M$ be the number of users, files and per-user cache memory size, \resp, and $t \eqdef  \frac{KM}{N}$  the (normalized) \agg cache size.
The well-known Ji-Caire-Molisch (JCM) caching scheme~\cite{ji2016fundamental} \achvs the \comm rate 
$\Rjcm\eqdef \frac{K(1-M/N)}{t}  $  with \sbp level (\ie, the number  of smaller packets  each file is split  into) $\Fjcm \eqdef t\binom{K}{t}$.
Compared to the MAN \schm which \achvs $\Rman\eqdef\frac{K(1-M/N)}{t+1}$ and $\Fman\eqdef \binom{K}{t}$, the \jsch has higher  \comm rate (at the same  per-user \memo size $M/N$) and \sbp: (i) $\Rjcm$ is larger than $\Rman$ \bcuz a smaller global caching gain $t$ is realized in the \jsch than $t+1$ in the \msch; (ii) $\Fjcm$ is $t$ times larger than $\Fman$  due to an additional layer of subfile splitting introduced to enable \tit{\sym} packet exchange~\cite[Definition 2]{zhang2026taming} during the multicast delivery phase, thereby achieving the optimal global caching gain $t$. 
\Msp, the
\jsch employs a two-layer  \sbp structure \af:
\begin{subequations}
\label{eq:jcm 2 layers,intro}
\begin{align}
\mrm{[First\; layer]}\;\; 
& W_n  = \left\{ W_{n, \Tc}\right\}_{\Tc  \in \binom{[K]}{t}}, \;\forall  n \label{eq:first layer, jcm spb,intro}\\
\mrm{[Second\;layer]}\;\;
&  W_{n, \Tc}  = \big \{ W_{ n, \Tc}^{(i)}    \big\}_{i\in [t]}, \;  \forall n, \Tc    \label{eq:2nd layer, jcm spb,intro}
\end{align}
\end{subequations}
\tit{1) File to subfiles (\ie, \man \Sbp):}
Each file is split into $\binom{K}{t}$ equal-sized \sbfs, each indexed by a $t$-subset of $[K]$ as shown in \eqref{eq:first layer, jcm spb,intro}.
In the cache placement phase,  
\user{k} stores all \sbfs where $k\in \Tc$, \ie,  
$Z_k=\{ W_{n, \Tc}: \forall  \Tc \ni k, \forall n \in [N]\}$. 
\tit{2) \Sbf to \pkts:} Each \sbf is further  split into $t$ equal-sized \pkts as shown in \eqref{eq:2nd layer, jcm spb,intro}.
\Aar, each file $W_n$ is split into $\Fjcm=t \binom{K}{t}$ \pkts, 
\be
\label{eq:jcm splitting,intro}
\mrm{[JCM\;splitting]}\;\;
W_n=\big\{W_{n, \Tc}^{(i)}\big\}_{\Tc \in  \binom{[K]}{t}, i \in [t]}. 
\ee

\subsection{\FinLen D2D \CoCa}
\label{subsec:fin len d2d,intro}
The large \sbp $\Fjcm=t\binom{K}{t}$ of JCM is a major obstacle to practical implementation, since it requires each file to be split into a prohibitively (\ie, exponential in $K$) large number of packets even for moderate  $K$.
As a result, reducing subpacketization has become a central problem in finite-length D2D coded caching~\cite{wang2017placement,chittoor2019low,8620232,9913463,9477627,woolsey2020d2d,wu2023coded,wang2025coded,nt2025d2d,rashid2026optimal,8437323,9448271,9103948}. 
Several constructions based on projective geometry~\cite{chittoor2019low,9477627}, D2D placement delivery arrays (DPDAs)\cite{wang2017placement,8620232,9913463,nt2025d2d,wang2025coded}, hypercube structures\cite{woolsey2020d2d,9448271}, and related combinatorial designs\cite{9103948} achieve reduced \sbp, often substantially below $\Fjcm$.
However, the \sbp \red  \achvd in 
these schemes comes inevitably at the \tit{cost of higher \comm rates} than the \jsch,  thus sacrificing rate optimality.
The only rate-optimal result established by Wang \etal\cite{wang2017placement,8620232} applies to a very limited set of \memo points
$ M/N \in\left\{\frac{1}{K},\frac{2}{K},\frac{K-2}{K},\frac{K-1}{K}\right\}$, and thus lacks generality.
Therefore, despite a rich literature on low-\sbp D2D caching, a general design methodology that simultaneously preserves the optimal rate and systematically reduces \sbp remains elusive.

\subsection{The Packet Type-based \Frmwk}
\label{subsec:pt framework,intro}
A significant step toward this goal is the packet type (PT)-based framework\cite{zhang2026taming}, which provides a \tit{systematic} way to design rate-optimal D2D coded caching schemes with reduced \sbp.
The key idea is to impose a  grouping on the user set and classify subfiles, packets, and multicast groups into types according to the induced geometric structure. This viewpoint reveals that the symmetric subpacketization and symmetric packet exchange used in the JCM scheme are not always necessary for achieving the optimal rate. In particular, the PT framework identifies two distinct sources of subpacketization reduction: 1) \tit{subfile saving gain}, obtained by excluding redundant subfile types, and 2) \tit{further splitting saving gain}, obtained by assigning fewer than $t$
packets to certain surviving subfile types through carefully coordinated transmitter selection. 
The joint  effect  of both \red gains leads to an improved \sbp over the \jsch in a wide range of parameter regimes while preserving  the same rate as $\Rjcm$. {\Ip,  three general classes of \dd \coca \schms  were constructed under the PT \frmwk:
(i) The first construction\cite[\Thm  1]{zhang2026taming} applies  to both even $K$ and even $t$ where $t\ge K/2$, and \achvs an order-wise (in $K$) \red over $\Fjcm$;
(ii) The second construction\cite[\Thm  2]{zhang2026taming} applies to both even $K$ and even $t$ and guarantees a more-than-half constant-factor \red;
(iii) The third construction\cite[\Thm 3]{zhang2026taming} applies to \params satisfying $K=mq$ where $m,q\ge t+1$, and \achvs  a constant-factor \red through  \sbf saving.}
\Aar, the PT framework revealed a sharper rate-\sbp trade-off boundary, that is, the \sbp of the \jsch can be reduced \tit{without} sacrificing the optimal rate.

However, the existing PT framework is restricted to \tit{homogeneous} \sbp, namely, all packets are required to have the same size regardless of their types.
This restriction is structural rather than cosmetic. Under unequal user grouping, the memory constraint of  the users depends not only on how many packets of each type are cached, but also on their \tit{sizes}. 
In certain parameter regimes, especially those involving odd 
$K$, homogeneous packet sizes make it impossible to balance the cache memories across users while retaining the desired packet reduction.
As a consequence, although the prior PT framework already covers several important regimes, \tit{it does not provide a rate-optimal construction for general  odd 
$K$}.
\Msp, although the third  construction in\cite{zhang2026taming} applies to  \tit{some} odd values of $K$ of the form  $K=mq$, where $m,q \ge t+1$, its scope is quite limited. 
Moreover, it exploits  only the  subfile saving gain and  does not explore  the \fssgain,  leaving room for further optimization.

\subsection{Summary of Contributions}
\label{subsec:summary of contribution,intro}

This work extends the PT framework from homogeneous to \tit{heterogeneous} subpacketization. By allowing different refined types to have different packet sizes, the proposed design introduces the flexibility needed to satisfy the memory constraint under unequal user grouping, thereby filling a gap left by existing homogeneous-subpacketization PT designs. Moreover, heterogeneous packet sizing is jointly designed with user grouping and transmitter selection, and thus becomes a key enabler of designing new low-subpacketization, rate-optimal \dd  \coca \schms.

The main contributions of this paper are  twofold:

1)  We extend the PT-based \frmwk~\cite{zhang2026taming} to the {\het} \sbp setting, allowing \pkts to have  type-dependent sizes
under a refined \tp  classification via \tit{\sbtps}. 
\Specly, each subfile type is further divided into two subtypes grouped into \tit{coupled groups}, and packets within the same coupled group share a common size while different groups may have distinct sizes. 
The proposed \frmwk jointly optimizes user \grpg, multicast transmitter selections with \diff \cpgrps, and
heterogeneous packet sizes are chosen to satisfy the users' cache   memory constraints while preserving the optimal \comm rate. 
\Ip, the \tx \selec in each \cpgrp yields an  {intermediate} \gfsv via the vector LCM operation,  and the sum of all \intmdt \gfsvs gives the  {\agg} \gfsv that determines the  final \sbp.
\Ip, the PT \sbp is an inner product between the \agg \gfsv  and  the subfile-count vector, enabling explicit characterization of \tit{subfile saving} (via exclusion of certain \sbf types) and \tit{further splitting saving} (via reduced FS factors).

2) Building on the \het PT framework, we construct a new class of \ropt D2D coded caching schemes for odd $ K$ and even $t$,
 achieving  a constant-factor \sbp reduction over the \jsch~\cite{ji2016fundamental}. 
The \ppsd design 
adopts an  unequal user \grpg $(\frac{K+1}{2}, \frac{K-1}{2})$ and
coordinated \tx \selecs across two \cpgrps.
The resulting intermediate FS vectors combine to form  an  \agg global FS vector that increases linearly with unit slope from zero and is capped at $t$. The packet sizes are then determined to satisfy the memory constraint, and the resulting \sbp ratio is characterized, including its asymptotic scaling in $K$ and $t$.
The proposed construction fills the gap in the \homo-\sbp PT design in~\cite{zhang2026taming}, where a general design for odd $ K$ is lacking.  It can thus be viewed as a complementary extension of~\cite{zhang2026taming}.

\emph{Organization.}
The remainder of  this paper is organized \af.
Section \ref{sec:problem setup and pt overview} describes the problem setup and 
provides an overview of the PT \frmwk.
Section  \ref{sec:main result} gives the main result and its interpretations. 
Section \ref{sec:het sbp pt,proposed frmwk} presents the proposed \het-\sbp \frmwk, and
Section \ref{sec:proposed construction} details the proposed construction. Section \ref{sec: discussion} provides several insightful discussions, and \Sec \ref{sec:conclusion} concludes the paper.

\emph{Notation.} 
$[m:n]\eqdef\{m,m+1,\cdots,n\}$, $[n]\eqdef [1:n]$.  
Calligraphic letters $\Ac,\Bc,\cdots$ denote sets. $\Ac\backslash \Bc\eqdef \{x\in \Ac:x\notin \Bc\}$, and $\binom{\Ac}{n}\eqdef \{ \Sc \subseteq \Ac: |\Sc| = n   \}$.  
Bold letters $\Am,\Bm, \av, \bv, \cdots $  denote matrices and vectors.
$\underline{a}_{n}\eqdef (a,\cdots,a)$ with $n$ identical entries. 
For $\av=(a_i)_{i=1}^n,\bv=(b_i)_{i=1}^n$, 
 $\av \preceq \bv$ means $a_i \le b_i, \forall i$.
For \randvars $X,Y$, $X\prec_{\rm st}$ (resp. $\preceq_{\rm st}$) denotes  strict (resp. non-strict) \stodom. 
$\oplus$ denotes the bit-wise XOR.
For nonnegative functions $f,g$, $f(n)=O(g(n))$ if $\exists C>0, n_0$ such that $f(n)\le Cg(n),\forall n\ge n_0 $; 
$f(n)=\Theta(g(n))$ if $\exists C_1,C_2>0,n_0$ such that $C_1\le f(n)/g(n)\le C_2,\forall n\ge n_0$.

\section{Problem  Setup and Overview of PT \Frmwk}
\label{sec:problem setup and pt overview}

\subsection{Problem Setup}
Consider a \dd \coca system with $N$ files $W_1, \cdots, W_N$, each of $L$ bits, and $K$ users, each equipped with a cache \memo $Z_k$ of $ML$ bits. 
In the placement phase, each user prestores uncoded file bits. 
After the users' demands 
$\dv \eqdef (d_1,\cdots, d_K)$ are revealed, 
delivery is performed over multicast groups of size $t+1$,  whose users take turns broadcasting coded messages to satisfy one another's demands.
To achieve the optimal delivery rate, each file must be partitioned into $F$ packets so as to realize the global caching gain $t$ (also referred to as {degrees-of-freedom} (DoF) as each transmitted bit is \simully useful to $t$ users).
Under uncoded placement and one-shot delivery, the \jcm rate 
$\Rjcm=\frac{K(1-M/N)}{t}$ is optimal for $N \ge  K$\cite{yapar2019optimality}.
The goal of this paper is to minimize the required subpacketization $F$ while preserving the optimal rate $\Rjcm$.

\subsection{Overview  of PT \Frmwk}
\label{subsec:overview of pt,proposed frmwk}
The key \compos of the \ptf\cite{zhang2026taming} include user grouping, subfile and packet types, multicast group types, \tx selection, 
local and global further splitting (FS) factors, the \vlcm (LCM) operation that merges \lfsfs into globally compatible ones, and the \memoconst (MC) verification.

\tit{1) User grouping:} The user set $[K]$ is partitioned into $m$ disjoint groups $\Qc_1, \cdots, \Qc_m$. A user grouping $\qv\eqdef  (q_1, \cdots, q_m)$, where $q_1\ge \cdots \ge  q_m>0$, is defined as the ordered representation of the sizes of the groups. $\qv$ is called an equal \grpg if $q_1=\cdots=q_m$, and \uneq  otherwise.
Let $N_{\rm d}(\qv) \le  m$ denote the number of distinct values in $\qv$, then \grps can be aggregated into \tit{unique sets}, representing unions of  user \grps with identical sizes.
Let $\hatu_i$ denote the $\ith$ \uset.

\tit{2) \Sbf, \pkt, and \mgrp types:} A \sbf $W_{n,\Tc}$ is a piece of $W_n$ that is exclusively stored by the $t$-user \sbst $\Tc$. A \mgrp $\Sc$ is a set of $t+1$ users. 
Given user \grpg $\qv$ with \grps 
$\Qc_1, \cdots, \Qc_m$, 
the \tp of $\Tc$ and (resp. $\Sc$) is defined as the ordered representation of the sizes of  the projections of $\Tc$ (resp. $\Sc$) onto $(\Qc_1, \cdots, \Qc_m)$.
The \tps of $\Tc$ and $\Sc$ are denoted by $\vv\eqdef (v_1, \cdots, v_m )$ and 
 $\sv\eqdef (s_1, \cdots, s_m )$, \resp, 
where $v_1 \ge \cdots \ge v_m \ge 0$,  
$s_1 \ge \cdots \ge s_m \ge 0$. 
The \sbf $W_{n, \Tc}$ and its support set $\Tc$ has the same \tp.
Given $\sv$, the \tit{\ists} $\Ic$ is defined as the set of \sbf \tps participating in the \txn within \tp-$\sv$ \mgrps.
Similarly, let  $\nd(\sv
)$ denote the number of distinct nonzero values in $\sv$.

\tit{3) \Tx \selec and local FS factors:}
The \fs (FS) splitting factor $\alpha(\vv) \in \mbb{Z}_+  $ for  \tp $\vv$ denotes the number of smaller \pkts into which each  \tp-$\vv$  \sbf is divided to ensure the optimal \dof of $t$. Given a \tp-$\sv$ \mgrp  $\Sc= \cup_{i\in  N_{\rm d}(\sv)} {\widehat{\Sc}_i}$, where $\widehat{\Sc}_i$ is the $\ith$ \uset, a subset of users $\Utx \subseteq \Sc$---which are restricted to be unions of \usets---are chosen as \txs, \ie,
$\Utx= \cup_{i\in \Dtx} \widehat{\Sc}_i$, $\Dtx  \subseteq [\nd(\sv)]$. 
The  \invod  \sbf \tps are $\Ic=\{\vv_i \}_{i=1 }^{N_{\rm d}(\sv) }   $, where $\vv_i$ denotes  the \tp of the \sbfs $\{W_{n, \Sc \bksl \{k\}   }   \}_{k \in  \widehat{\Sc}_i }$.
The FS factor is then given by
\begin{eqnarray}
\label{eq:local FS factor,pt overview}
\alpha(\vv_i)=\left\{\begin{array}{ll}
\sum\limits_{k\in\mathcal{D}_{\rm Tx}}|\widehat{\Sc}_k|  -1,& \text{if}\; i\in\mathcal{D}_{\rm Tx}\\
\sum\limits_{k\in\mathcal{D}_{\rm Tx}}|\widehat{\Sc}_k|, &\text{if} \; i\notin\mathcal{D}_{\rm Tx}
\end{array}\right.
\end{eqnarray}
Hence,  $\alpha(\vv_i)$
equals the number of \txs observable by each user in $ \widehat{\Sc}_i$, excluding self\txn.
The FS factors obtained from  \eqref{eq:local FS factor,pt overview} are called \tit{local} FS factors. 
Note that the same \tx \selec is applied  to all \mgrps of the same \tp. Moreover, the \jsch is a special PT design using
$\Utx=\Sc, \forall  \Sc$, yielding uniform FS factors  $\alpha(\vv_i)=t,\forall i$.

\tit{4) Vector LCM and global FS factors:}
\Tx \selec needs to be determined  for all \mgrp \tps. Since one \sbf \tp may be \invod in more than one \tp of \mgrps, the \tx \selecs in them may yield distinct FS factors for the same \sbf \tp. 
Hence, the local FS factors obtained in \eqref{eq:local FS factor,pt overview} are coordinated through the \vlcm (LCM) operation\cite{zhang2026taming} to generate  a set of \tit{global} FS factors that are compatible with all \lfsfs. 
Let $ \alphaglobal \eqdef [ \alpha^{\rm global}(i)]_{i=1}^V$ ($V$ being the total number of \sbf \tps) denote the \gfsv, where each \tp-$\vv_i$ \sbf is split into $\alpha^{\rm global}(i)$ \pkts, 
\ie, $W_{n, \Tc}=\{ W_{n,\Tc}^{(k)}\}_{k=1}^{\alpha^{\rm global}(i)} $.
Then, the overall file splitting\footnote{Suppose the \memoconst in \eqref{eq:mc,pt overview} is satisfied.} under PT is 
\be 
\label{eq:pt file splitting,pt overview}
W_n = \bigcup\nolimits_{i \in [V]  }
\big\{ W_{n,  \Tc}^{(k)}       \big\}_{\forall  \Tc, \type(\Tc)=\vv_i,k \in  [\alpha^{\rm global}(i)] },\; \forall n
\ee 
and the \sbp level is equal to
\be
\label{eq:Fpt,homo pt}
\Fpt= \alphaglobal\Fm^{\rm T}
=
\sum\nolimits_{i=1}^V  \alpha^{\rm global}(i)F(\vv_i), 
\ee 
where 
$F(\vv_i)$ denotes the number of \tp-$\vv_i$ \sbfs, and  $\Fm \eqdef [F(\vv_i)]_{i=1}^V$ denotes the \sbf number vector.

\tit{5) \Memoconst verification:}
Under \homo-\sbp, all \pkts have an identical  size. Hence, we must ensure all users store the same number of \pkts in the placement phase.
Let $\Fm_i \eqdef [F_i(\vv_j)]_{j=1}^V  $ be the \tit{user cache vector}, which represents the number of \sbfs of each \tp stored by any user in the $\ith$ \uset of $\qv$.
Also let  $\Deltam_i \eqdef  \Fm_{i+1} - \Fm_i, i \in  [\nd(\qv)-1]   $ be  the \tit{cache \diffce  vector},  the $\jth$ entry of which $ \Delta_i(j) =F_{i+1}(\vv_j) -F_{i}(\vv_j),j\in  [V]$ represents the \diffce  in  the number of \tp-$\vv_j$ \sbfs stored per user in $\hatu_{i+1}$ and $\hatu_{i}$. 
\Aar, $\alphaglobal$ must satisfy the
\memoconst (MC): 
\be 
\label{eq:mc,pt overview}
\alphaglobal \Deltam_i^{\rm T}=0, \; \forall i\in   [\nd(\qv)-1],
\ee 
which is equivalent to  $\alphaglobal\Fm_1^{\rm T}= \alphaglobal\Fm_i^{\rm T}, \forall i \in [\nd(\qv)]$, indeed ensuring all users store the same number of  \pkts. Note that  the size of the packets can be  tuned to satisfy $|Z_k|=ML, \forall k\in [K]$. Also note  that under equal user \grpg, there is only one \uset so that the MC is automatically satisfied due to symmetry.

\section{Main Result}
\label{sec:main result}

\begin{theorem}
\label{thm:thm}
Let $ (K,t) =(2q+1,2r)$ with $q,r \in \mbb{N}_+$. The \comm  rate $R=N/M-1$ of D2D coded caching is achievable with \sbp level $\Fpt$ satisfying $\Fpt \le \Fjcm$, where
\be
\label{eq:Fpt, thm}
\Fpt= \sum\nolimits_{k=1}^{t+1}\alpha_kf_k,
\ee 
with
\begin{subequations}
\label{eq: alpha_k & f_k, thm}
\begin{align}
\alpha_k &  =
\begin{cases}
2(k-1), &    \mrm{if}\;  k\in [r],    \label{eq:alpha_k, thm}\vspace{-.15cm}
\\
t, & \mrm{if}\; k\in [r+1:t+1],
\end{cases}
\\
f_k & =\binom{q+1}{k-1}\binom{q}{t-k+1}, \;  k\in [t+1]. \label{eq:f_k, thm}
\end{align}
\end{subequations}
Moreover, for fixed $t$,
\be
\label{eq:asymp bound,thm}
\lim_{q\to \infty} \frac{F_{\rm PT}}{F_{\rm JCM}} = 1 -  \frac{1}{2}\cdot\frac{\binom{t}{r}}{2^t} = 1 - \Theta\lef( \frac{1}{\sqrt{t}}  \rig).
\ee 
\end{theorem}

The proof of \Thm \ref{thm:thm} is  given in \Sec \ref{sec:proposed construction}, where we present a novel two-\grp  unequal  user \grpg   \PTd   that \achvs a constant-factor \sbp \red \relato the \jsch. 
Fig.~\ref{fig:sbp ratios,main thm} shows a comparison of the
actual \sbp ratios $\Fpt/\Fjcm$ (see (\ref{eq:Fpt, thm})) with  the corresponding \asympt limits (see (\ref{eq:asymp bound,thm}))
for various $t$. The actual and \asympt ratios match closely.
\begin{figure}[t]
\centering
\includegraphics[width=0.38\textwidth]{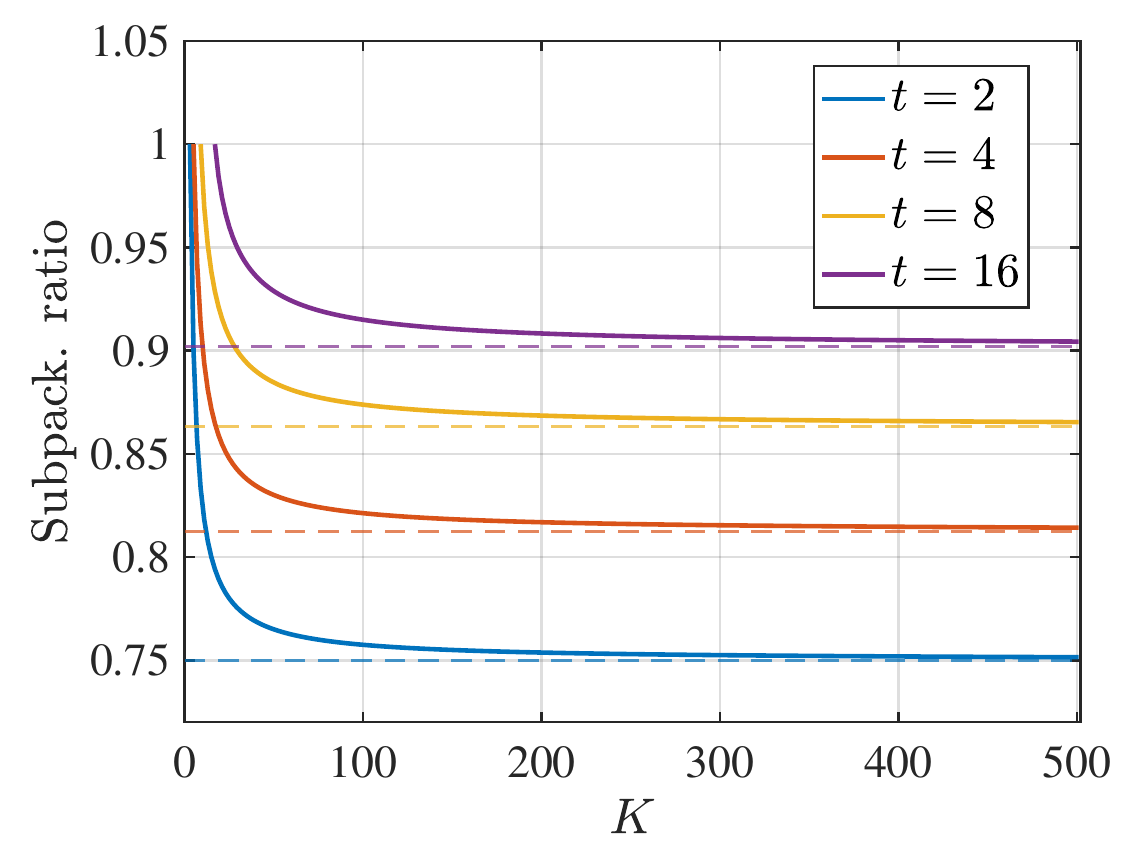}
\vspace{-.2cm}
\caption{\small Actual (solid)  and \asympt (dashed) \sbp ratios for various $t$ values.
Both ratios \incrs with $t$, 
indicating a diminishing PT reduction over JCM.}
\label{fig:sbp ratios,main thm}
\end{figure}

We highlight the implications of \Thm \ref{thm:thm} \af:

\tit{1) \Sbp Reduction over JCM~\cite{ji2016fundamental}:} The PT \sbp (\ref{eq:Fpt, thm}) equals the 
inner product of the \seqs $\{\alpha_k\}_{k=1}^{t+1}$ and $ \{f_k\}_{k=1}^{t+1}$, where $f_k$ denotes  the number of \sbfs  of the $\kth$ \tp and $\alpha_k $ denotes the  \corrspdg \fsf (\ie, the number of \pkts per \sbf).  In contrast, the JCM \sbp can be written as
\be
\label{eq:Fjcm,thm implication}
\Fjcm  = \sum\nolimits_{k=1}^{t+1} t f_k,
\ee
where  $\sum_{k=1}^{t+1}f_k=\binom{K}{t}$\footnote{See explanation in \Rmk \ref{rmk:JCM sbp identity}.}. 
Since  $\alpha_k \le t $  for all $k$, we have $\Fpt \le \Fjcm$, \ie, the proposed \PTd \achvs (strictly) lower \sbp than the \jsch while preserving the same \comm rate. 
More compactly, (\ref{eq:alpha_k, thm}) corresponds to the \gfsv 
\be 
\label{eq:alpha global,thm implication}
\alphaglobal=\lef(0,2,4,\cdots,t-2, \underline{t}_{r+1} \rig)
\ee 
Comparing $\Fpt=\alphaglobal f_{1:k}^{\rm T}  $ with  $ \Fjcm= \underline{t}_{t+1}f_{1:k}^{\rm T}$, it can be seen that
the \red of PT arises from two effects:  \sbf exclusion (as $\alpha_1=0$, \pkts of  the first \tp are eliminated) and smaller FS factors (as $\alpha_k< t ,\forall k \le r$). These effects correspond to the \tit{\ssgain} and \tit{\fssgain} identified  in the PT \frmwk~\cite{zhang2026taming}.

\tit{2) Asymptotic Scaling:}
For fixed $t=K \mu$, the \asympt ratio (\ref{eq:asymp bound,thm}) has a $1-\Theta(1/\sqrt{t})$ scaling in the small \memo regime $\mu=\Theta(1/K)$. This implies a constant-factor \red over JCM.
Fig. \ref{fig:sbp ratios,main thm} shows that for any $t$,  the actual ratio \decrs with $K$,\footnote{See proof in Lemma \ref{lemma:monotonicity of Fpt/Fjcm}, \Sec \ref{subsec:subpkt analysis, gen pt}.} and both the actual and \asympt ratios \incrs with $t$, indicating 
a weaker \sbp \red for larger $t$.

\tit{3) Comparison with \Homo PT Design~\cite{zhang2026taming}:} 
For even $K$  and $t$, the \homo-\sbp \PTd
in~\cite[\Thm 2]{zhang2026taming}
\achvs a more-than-half \sbp \red compared with  the \jsch. 
However, that construction does not extend to odd $t$, neither does~\cite{zhang2026taming} provide a design for odd $K$  \ingen.
Therefore, our design in Theorem \ref{thm:thm} serves as a complementary solution to~\cite{zhang2026taming}.

\tit{4) The Case  of Even  $K$:}
Although Theorem~\ref{thm:thm} is stated for odd $K$, it extends to even $K$ with minor modifications.
\Specly, for $(K,t)=(2q,2r)$  with $q\ge 2r+1$, under the user \grpg
$\qv=(q+1,q-1)$ and the \tx \selecs in (\ref{eq:tx select G1, gen PT}) and  (\ref{eq:tx select G2, gen PT}), 
the \gfsv
$\alphaglobal =(\underline{t}_{r+1},t-2,t-4,\cdots,2,0)$ is \achvb. 
However, the more-than-half \red established in~\cite[Theorem~2]{zhang2026taming}  cannot not be guaranteed in this setting.
Consequently, when both $K$ and $t$ are even, the homogeneous PT design in~\cite{zhang2026taming} remains preferable.

\section{PT \Frmwk with \Het \Sbp}
\label{sec:het sbp pt,proposed frmwk}

The \homo-\sbp assumption, \ie, requiring all packets to have the same size regardless of their types, severely limits the applicability of the \ptf\cite{zhang2026taming}. \Fex, that framework does not yield good PT designs for general odd $K$ due to the challenge in satisfying the \memoconst \eqref{eq:mc,pt overview}\footnote{Although \cite{zhang2026taming} proposed a PT design that \achvs an order-wise \sbp \red for $(K,t)=(2m+1,2m-1), m\ge 3$,
the result is confined to the special case $t=K-2$ and is therefore not general.}.
To fill the gap, 
we propose a new \het-\sbp \ptf
where
packets associated with different  types can have different sizes, while preserving a common type-based design philosophy.
The added flexibility in choosing the \pkt sizes 
helps satisfy the users' \memoconst while enabling meaningful \pkt \red.
Concretely, each subfile type is further refined into \tit{subtypes}, which are organized into \tit{coupled groups}. Packets within the same coupled group share a common size, whereas packets from different coupled groups may have different sizes. This additional degree of freedom makes it possible to satisfy the memory constraint under unequal user grouping even when no homogeneous design can do so. At the same time, the transmitter selections across coupled groups are coordinated so that their induced intermediate FS vectors combine into an aggregate global FS vector, which in turn determines the final subpacketization level. In this way, heterogeneous packet sizing is not an afterthought, but an integral design variable that works jointly with user grouping and transmitter selection. 

\Itf, we describe the modifications to the original  \ptf  to allow \het \sbp.

\tit{1) Subfile subtypes,  \pkt \tps, and \cpgrps:} 
In the \orig \ptf, each \sbf  $W_{n,\Tc}$  and its \pkts $\{W_{n,\Tc}^{(k)}\}_{k =1}^{\alpha^{\rm global}(i) }  $ have the same \tp $\vv_i, i \in [V]$. 
In  the \het \ptf, we refine the classification by dividing  each  \sbf \tp $\vv_i$ into $G$ \tit{\sbtps} $\vv_i^{(1)}, \cdots, \vv_i^{(G)}$, each \corrspdg to a \diff \pkt \tp. 
\Pkt \tps with the same superscript constitute a \tit{\cpgrp}, \ie, the $\gth$ of which is denoted by 
\be
\Gc_g \eqdef  \big\{ \vv_i^{(g)}\big\}_{i\in  [V]}, \;g\in [G]
\ee 
\Pkts of \tps from  the same $\Gc_g$ have an identical size of $\ell^{(g)}$ bits, whereas  \pkts from  \diff \cpgrps  may have distinct sizes.

\tit{2) \Tx \selec \& global FS factors:}
\Tx \selec is  performed \indeptly over each \cpgrp. Let $\alphabm^{(g)} \eqdef [\alpha^{(g)}(i)]_{i=1}^V  $  denote the (intermediate) \gfsv derived from the $\gth$  \tx \selec  for $\Gc_g, g\in  [G]$. 
The \agg \gfsv  is then equal to
\be
\label{eq:overall glb fs vector,het pt}
\alphaglobal = \sum\nolimits_{g=1}^G \alphabm^{(g)},
\ee 
and the overall \sbp level is given by 
$\Fpt= \alphaglobal \Fm^{\rm T} $, which is similar to \eqref{eq:Fpt,homo pt}.
Compared with the \homo PT setting in \Sec \ref{subsec:overview of pt,proposed frmwk}, where \tx \selec is performed only once, the proposed \het setting allows $G$ separate \tx \selecs. Combined with nonuniform \pkt sizing, this added flexibility helps satisfy the \memoconst of the users.

\tit{3) \Memoconst (MC):} 
Define $\gamma_g  \eqdef \ellg{g}/\ell^{(1)}, g\in[G]$ as the  \pkt size ratio between the $\gth$  and the first \cpgrps (note that  $\gamma_1=1$).
The total number of bits stored per file by any user in the $\ith$ \uset  $\hatu_i$ is equal to 
$
\sum_{g=1}^G\sum_{j=1}^V F_i(\vv_j)\alpha^{(g)}(j) \ellg{g}=  \sum_{g=1}^G  \ellg{g}\alphabm^{(g)}\Fm_i^{\rm T}, i \in  [\nd(\qv)]
$.
To ensure users across \diff \usets store the same number of bits, we require
$ \sum_{g=1}^G  \ellg{g}\alphabm^{(g)}\Fm_i^{\rm T} = \sum_{g=1}^G  \ellg{g}\alphabm^{(g)}\Fm_{i+1}^{\rm T}$ for all $i\in [\nd(\qv)-1]$, which is \eqvlt to 
$ \ellg{1} \sum_{g=1}^G \gamma_g \alphabm^{(g)} \Deltam_i^{\rm T}=0$ (recall that  $\Deltam_i = \Fm_{i+1}- \Fm_i$). Since $\ellg{1}>0$, the \memoconst can be expressed as
\be 
\label{eq:MC,het pt}
\sum\nolimits_{g=1}^G \gamma_g \alphabm^{(g)} \Deltam_i^{\rm T}=0,\; \forall i\in [\nd(\qv)-1]
\ee
Given the \tx \selecs for the $G$ \cpgrps and the \corrspdg \gfsvs $\alphabmg{1},\cdots ,\alphabmg{G}$,  the \pkt size ratios $\gammag{1}, \cdots, \gammag{G}$ can be solved. The specific \pkt sizes $\ellg{1},\cdots, \ellg{G}$ can then be determined by the per-user storage \req
$N \ellg{1}\sum\nolimits_{g=1}^G\gammag{g} \alphabm^{(g)} \Fm_i^{\rm T}=ML$ (using any $i$),
yielding 
\be 
\label{eq:pkt sizes,het pt}
\ellg{1}=
\frac{ML}{N \sum\nolimits_{g=1}^G\gammag{g} \alphabm^{(g)} \Fm_i^{\rm T}}, \; 
\ellg{g}= \gammag{g}\ellg{1}, g \in[2:G]
\ee

\begin{remark}
\label{rmk:homo pt as specical case of het pt}
\tit{Note that when $G=1$, \eqref{eq:MC,het pt} reduces to \eqref{eq:mc,pt overview} since $\gamma_1=1$. Hence, the  \orig \homo PT can be viewed as a special case of our \het \ptf.
\Inadd, \eqref{eq:MC,het pt} is a linear system with $G-1$ variables $\gammag{2}, \cdots, \gammag{G}$ ($\gammag{1}=1$) and $\nd(\qv)-1$ equations.  To  ensure existence of solutions, 
the number of equations should be  no more than the number of variables, \ie, $G \ge \nd(\qv)  $, implying that the number of  \usets  in the user \grpg should not exceed  the number of
\cpgrps.
}\end{remark}

\tit{4) Validation of \pkt size ratios:}
The solution of \eqref{eq:MC,het pt} must be positive numbers, \ie, $\gammag{g}>0,\forall g\in [G]$,
as they represent ratios of \pkt sizes. 
This must be enforced given the user \grpg and \tx \selec.

\Insum, the design procedure under  the \het \ptf  consists of  the following steps:

\tit{Step 1:}  \Dtm  user \grpg $\qv$ and \sbf types $\{\vv_i\}_{i=1}^V$;

 \tit{Step 2:}  \Dtm \pkt \tps (\ie, \sbf  \sbtps), and \cpgrps $\Gcg=\{ \vv_i^{(g)}\}_{i=1}^V, g\in [G]$;

\tit{Step 3:} \Tx \selec for each $\Gcg$, and intermediate \gfsvs $ \alphabm^{(g)},g\in[G] $; 

 \tit{Step 4:} \Dtm \pkt sizes $\ellg{g}, g\in [G]$ \accrdto \eqref{eq:MC,het pt} and \eqref{eq:pkt sizes,het pt};

 \tit{Step 5:} Verify the validity of the computed \pkt sizes (they must be positive integers); 

 \tit{Step 6:} If {Step 5} is satisfied, output the  caching  \schm with \sbp structure determined by $\{\alphabmg{g},\ellg{g}\}_{g\in  [G]}$. The overall \sbp level is equal to
\be
\label{eq:sbp,het pt}
\Fpt= \alphaglobal\FmT=\sum\nolimits_{g=1}^G \alphabmg{g}\FmT.
\ee

\section{Proposed Construction}
\label{sec:proposed construction}

This \sect presents the  proposed \PTd that \achvs the \sbp in \Thm \ref{thm:thm}. For  $K=2q+1$, the proposed \constrtn uses
a two-\grp \uneq user  \grpg  $\qv=(q+1, q)$ with $G=2$ \cpgrps. The \mtcst \txs are carefully selected for  each \cpgrp \soth the two intermediate \gfsvs $\alphabmg{1}$ and  $\alphabmg{2}$  combine into an overall $\alphaglobal$ that yields both \ssgain and \fssgain.

We begin with an illustrative example to highlight the core design principles.

\subsection{Motivating Example}
\label{subsec:examples, proposed construction}

\begin{example}[File Splitting, $t=2$]
\label{example:t=2} 
Consider $(K,t)=(2q+1, 2)$ with $q\ge 3$. Under the user grouping $\qv=(q+1,q)$ with specific assignment $\Qc_1=\{1, \cdots, q+1 \}  $ and $\Qc_2=\{q+2, \cdots, K \}  $, there are $V=3$ \sbfts and $S=4$ \mgrpts, 
\begin{align}
& \vv_1=(0,2), \;  \vv_2=(1,1), \; \vv_3=(2,0),  \notag\\
& \sv_1=(0,3),\;\sv_2=(1,2),\;\sv_3=(2,1),\; \sv_4=(3,0),
\end{align}
with \invod  \sbfts $\Ic_1=\{\vv_1\}  $, $\Ic_2=\{\vv_1, \vv_2\}  $, $\Ic_3=\{\vv_2,\vv_3 \}  $, and $\Ic_4=\{\vv_3\}$. 
We divide each \sbft  $\vv_i$ into two \tit{subtypes}, denoted by $\vv_i^{(1)}$ and $\vv_i^{(2)}, i\in [3]$, \resp,  which will be stored by the same set of users in the cache placement phase.
Define the \tit{coupled group}, denoted by $\Gc$, as the set of \sbtps which have the same superscript. Here, we have two \cpgrps  $\Gc_k=\{ \vv_i^{(k)} \}_{i=1}^3,k=1,2$. 
\Pkt of the \sbtps in each  \cpgrp $\Gc_k$ have the same size (in bits), denoted by $\ellg{k},k=1,2$.
Different transmitter selection strategies can be assigned not only across multicast group types, but also across the two subtypes of each subfile type involved in a given group type, resulting in two distinct global FS vectors. 
Moreover, the number
of subfiles of each type and its two subtypes are identical, \ie, $F(\vv_i)= F(\vv_i^{(1)})= F(\vv_i^{(2)}), i\in [3]$. The user  cache vector is
\begin{align}
\Fm_1 & \eqdef  \lef [F_1(\vv_1), F_1(\vv_2), F_1(\vv_3)  \rig   ]=\lef( 0, q,q  \rig), \notag\\
\Fm_2 & \eqdef \lef [F_2(\vv_1), F_2(\vv_2), F_2(\vv_3)  \rig   ]=\lef( q-1, q+1,0 \rig),
\end{align}
yielding cache \diffce vector 
\be
\label{eq:Delta,example}
\bm{\Delta}=\Fm_2-\Fm_1=(q-1,1,-q).
\ee 

For \sbtps in each \cpgrp, we employ a \diff \tx selection \af:
\begin{align}
\Gc_1:  \sv_1=(0,3^\dagger), \sv_2=(1^\dagger,2),\sv_3=(2^\dagger,1), \sv_4=(3^\dagger,0), \label{eq:tx select in G1}  \\
\Gc_2:  \sv_1=(0,3^\dagger), \sv_2=(1^\dagger,2),\sv_3=(2, 1^\dagger), \sv_4=(3^\dagger,0). \label{eq:tx select in G2}
\end{align}
The intermediate \gfsvs resulting from each \cpgrp  can be  computed as (see Table \ref{tab:FS_tables, example})
\begin{align}
\alphabm^{(1)} & \eqdef \big [\alpha(\vv_1^{(1)}), \alpha(\vv_2^{(1)}),\alpha(\vv_3^{(1)})\big]
=(0,1,2),   
\label{eq:intermediate gfsv (1), example}\\
\alphabm^{(2)} & \eqdef \big [\alpha(\vv_1^{(2)}), \alpha(\vv_2^{(2)}),\alpha(\vv_3^{(2)})\big]
=(0,1,0).\label{eq:intermediate gfsv (2), example}
\end{align}
\begin{table}[t]
\centering
\caption{\small FS Tables for $\mathcal{G}_1$ and $\mathcal{G}_2$}
\label{tab:FS_tables, example}
\setlength{\tabcolsep}{5pt}
\renewcommand{\arraystretch}{0.8}
\begin{subtable}[t]{0.48\columnwidth}
\centering
\caption{$\mathcal{G}_1$}
\begin{tabular}{|c|c|c|c|}
\hline
 & $\vv_1^{(1)}$ & $\vv_2^{(1)}$ & $\vv_3^{(1)}$ \\
\thickhline
$\sv_1$ & $2 $&  $\star$  &  $\star$ \\
\hline
$\sv_2$ & $0$ & $1 $& $\star$  \\
\hline
$\sv_3$ & $\star$  & $1$ & $2$ \\
\hline
$\sv_4$ &  $\star$ & $\star$  & $2$ \\
\thickhline
$\alphabm^{(1)}$ & $0$ & $1$ & $2$ \\
\hline
\end{tabular}
\end{subtable}
\hfill
\begin{subtable}[t]{0.48\columnwidth}
\centering
\caption{$\mathcal{G}_2$}
\begin{tabular}{|c|c|c|c|}
\hline
 & $\vv_1^{(2)}$ & $\vv_2^{(2)}$ & $\vv_3^{(2)}$ \\
\thickhline
$\sv_1$ & $2$ & $\star$  & $\star$  \\
\hline
$\sv_2$ &$ 0$ & $1$ &  $\star$ \\
\hline
$\sv_3$ & $\star$  & $4 $& $0$ \\
\hline
$\sv_4$ &  $\star$ & $\star$  & $2$ \\
\thickhline
$\alphabm^{(2)}$ & $0$ & $1$ & $0$  \\
\hline
\end{tabular}
\end{subtable}
\end{table}
Clearly, when considering only \sbtps in $\Gc_1$  or $\Gc_2$, the \memoconst cannot be satisfied (assuming \sbfs in each  $\Gc_k$ have an identical size):
\begin{align}
\label{eq:MC not satisfied for each cpg,example}
\alphabm^{(1)} \bm{\Delta}^{\rm T} &  =(0,1,2)(q-1,1,1-q)^{\rm T}=1-2q \ne 0,\notag\\
\alphabm^{(2)} \bm{\Delta}^{\rm T} & =(0,1,0)(q-1,1,1-q)^{\rm T}=1 \ne 0.
\end{align}

To satisfy the MC, we employ both \sbtps of \sbfs, and assume that \pkts from \diff \sbtps can have distinct sizes. Let the overall \gfsv be
\be
\label{eq:alpha-global, example}
\alphaglobal = \alphabm^{(1)} + \alphabm^{(2)}=(0,2,2),
\ee 
implying that type-$\vv_1$ (and its two \sbtps $\vv_1^{(1)}, \vv_1^{(2)}$) \sbfs are excluded.
Each type-$\vv_i$ subfile, for $i=2,3$, is split into two packets of types $\vv_i^{(1)}$ and $\vv_i^{(2)}$, with sizes $\ell^{(1)}$ and $\ell^{(2)}$ bits, respectively.  
Since all users in the same user \grp will store the same number of bits due to the symmetry of cache placement, we need  to ensure that users in $\Qc_1$ and $\Qc_2$  store the same number of bits; that is, 
\be
\alphabm^{(1)}\Fm_1^{\rm T}\ell^{(1)} + 
\alphabm^{(2)}\Fm_1^{\rm T}\ell^{(2)} = 
\alphabm^{(1)}\Fm_2^{\rm T}\ell^{(1)} + 
\alphabm^{(2)}\Fm_2^{\rm T}\ell^{(2)}.
\ee
which is equivalent to 
$\alphabm^{(1)} \Deltam^{\rm T} \ell^{(1)} +
\alphabm^{(2)} \Deltam^{\rm T} \ell^{(2)}=0
$.
Thus, the \pkt  size ratio is equal to 
\be
\label{eq:gamma, example}
\gamma\eqdef \frac{\ell^{(2)}}{\ell^{(1)}} =
- \frac{\alphabm^{(1)} \Deltam^{\rm T}}{\alphabm^{(2)} \Deltam^{\rm T}}=-
\frac{(0,1,2)(q-1,1,-q)^{\rm T} }{(0,1,0)(q-1,1,-q)^{\rm T}}=K-2.
\ee 
Since each file contains $L$  bits  in total, 
\be
\label{eq:L=alpha1*F*l1+alpha1*F*l1, example}
\alphabm^{(1)}\Fm^{\rm T}\ell^{(1)} + 
\alphabm^{(2)}\Fm^{\rm T}\ell^{(2)}=L, 
\ee
where $\Fm=[F(\vv_1),F(\vv_2), F(\vv_3)]=\big[ \binom{q}{2}, q(q+1),  \binom{q+1}{2} \big]$ is the \sbf number vector. 
Plugging $ \ell^{(2)}= \gamma\ell^{(1)} $  into (\ref{eq:L=alpha1*F*l1+alpha1*F*l1, example}),  the \pkt sizes can be solved as 
\begin{align}
\label{eq:pkt size,example}
\ell^{(1)} & = \frac{L}{(\alphabm^{(1)} + \gamma\alphabm^{(2)})\Fm^{\rm T}} = \frac{4L}{K(K^2-1)}, \notag \\
 \ell^{(2)} & =\gamma\ell^{(1)}= \frac{4(K-2)L}{K(K^2-1)}.
\end{align}

The overall \sbp is determined by the \gfsv (\ref{eq:alpha-global, example}),
\begin{subequations}
\label{eq:Fpt calculation,example}
\begin{align}
\label{eq:Fpt, example}
 \Fpt & = \sum\nolimits_{i=1}^2  \alphabm^{(i)}
\big [F(\vv_1^{(i)}), F(\vv_2^{(i)}),F(\vv_3^{(i)} )   \big]^{\rm T} \\ 
& = \big( \alphabm^{(1)} + \alphabm^{(2)}   \big)
\big [F(\vv_1 ), F(\vv_2 ),F(\vv_3 )   \big]^{\rm T}\label{eq:step 1,Fpt,ex}\\
&\overset{(\ref{eq:alpha-global, example})}{=}  \alphaglobal  \Fm^{ \rm T} \\
& = (0,2,2)\lef[ \binom{q}{2}, q(q+1),  \binom{q+1}{2} \rig]^{\rm T}= \frac{3(K^2-1)}{4},
\end{align}
\end{subequations}
where  (\ref{eq:step 1,Fpt,ex}) is \bcuz  $F(\vv_k)=F(\vv_k^{(i)}), \forall k\in [3], i\in[2]$, \ie, the number of \sbfs of each type is  equal to its  \sbtps. 
Hence, $ \Fpt/\Fjcm=\frac{3}{4}(1+\frac{1}{K}) \le \frac{6}{7}, \forall K \ge  7 (q\ge 3)$, and $ \lim_{K\to \infty} \Fpt/\Fjcm=3/4$, 
implying a $1/4$  \asympt \sbp  \red compared to JCM\footnote{Since $\mu= t/K$, letting $K\to \infty$ while fixing  $t=2$ implies $\mu  \to 0$.}.  
\begin{figure}[t]
\centering
\includegraphics[width=0.38\textwidth, height=3cm]{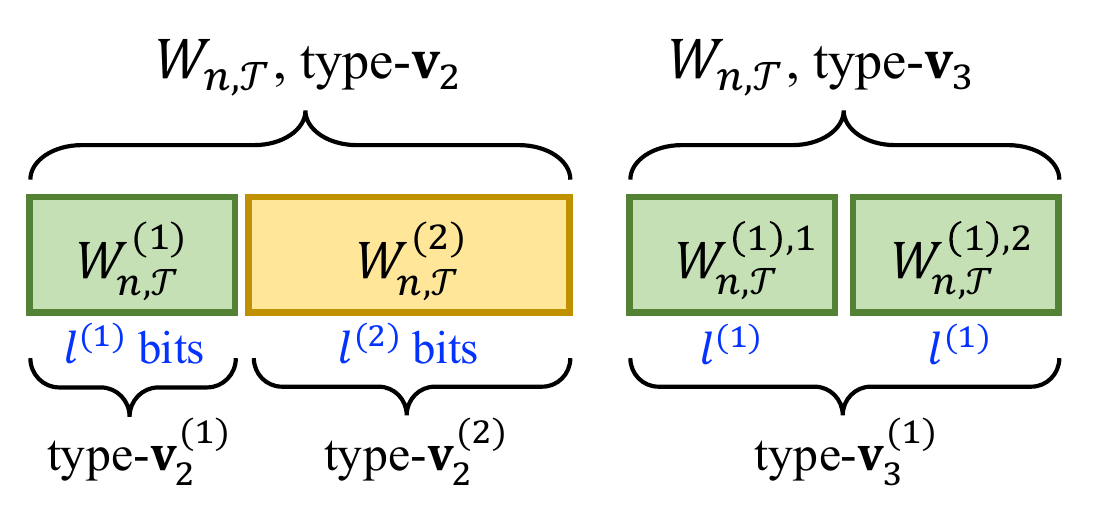}
\vspace{-.2cm}
\caption{\small Illustration of subfile and \pkt types. \Pkts  of the same color (\ie, belonging to the same \cpgrp) have identical sizes. }
\label{fig:subfile & pkt types, example}
\end{figure}

\tbf{File splitting.}  
The file splitting is determined by \tit{both} the \intmdt and \gfsvs (\ref{eq:intermediate gfsv (1), example}), (\ref{eq:intermediate gfsv (2), example}), and (\ref{eq:alpha-global, example}). \Ip, type-$\vv_1$ \sbfs and  type-$\vv_3^{(2)}$ \pkts are excluded.  For the remaining types, the FS factors are summarized in Table~\ref{tab:FS factors,example}.
\begin{table}[ht]
\caption{\small Summary of FS Factors, Example \ref{example:t=2}}
\label{tab:FS factors,example}
\vspace{-.05cm}
\centering
\setlength{\tabcolsep}{2pt}
\renewcommand{\arraystretch}{0.85}
\begin{tabular}{|c|c|c|}
\hline
\small \Sbf \tp   & \small \Pkt \tp  & \small FS factor  \\ \thickhline
$\vv_1$ & \small N/A & $0$ \small (excluded) \\ \hline
\multirow{2}{*}{ $\vv_2$ } & $\vv_2^{(1)}$ &  $1$ \\ \cline{2-3}
&  $\vv_2^{(2)}$ & $1$ \\ \hline
\multirow{2}{*}{$\vv_3$} &  $\vv_3^{(1)}$ & $2$ \\ \cline{2-3}
& \small N/A & $0$ \small (excluded)  \\ 
\hline
\end{tabular}
\end{table}
The file splitting is \af: 

i) \tit{File to \sbf:} Each file is  partitioned into $F(\vv_2) $ type-$\vv_2$ \sbfs and $F(\vv_3)$ type-$\vv_3$ \sbfs.
A \tp-$\vv_2 $ \sbf  has size $\ell^{(1)} +\ell^{(2)} $ bits, while a \tp-$\vv_3$ \sbf has $2\ell^{(1)}$ bits. Each \sbf \tp $\vv_2$ has two  \sbtps $\vv_2^{(1)}$ and $\vv_2^{(2)}$,  and each  \sbf \tp  $\vv_3$ has only one \sbtp $ \vv_3^{(1)}$.

ii) \tit{\Sbf to \pkts:} 
Each type-$\vv_2$ \sbf  is split into one  \tp-$\vv_2^{(1)}$ \pkt and one \tp-$\vv_2^{(2)}$ \pkt, \ie, 
$W_{n,\Tc}=\big(W_{n,\Tc}^{(1)}, W_{n,\Tc}^{(2)}     \big)$  if $\type(\Tc)=\vv_2$.    
Each type-$\vv_3$ \sbf  is split into two \tp-$\vv_3^{(1)}$ \pkts, \ie, 
$W_{n,\Tc}=\big(W_{n,\Tc}^{(1),1}, W_{n,\Tc}^{(1),2}     \big)$  if $\type(\Tc)=\vv_3$. 
For a \tp-$\vv_k$ subset $\Tc$, $W_{n,\Tc}^{(i),j}$ ($i \in [2],j\in [\alpha(\vv_k^{(i)})]$) denotes the $\jth$  \pkt of \tp-$\vv_k^{(i)}$. 
See Fig. \ref{fig:subfile & pkt types, example}  for an \illutn  of   \sbf composition.

\Aar, each file can be expressed as
\begin{align}
W_n = \Pc_2^{(1)} \cup  \Pc_2^{(2)}   \cup  \Pc_3^{(1)},
\end{align}
where $\Pc_k^{(i)}$ denote  the set  of all  \tp-$\vv_k^{(i)}$ \pkts. \Specly,
\begin{align}
\label{eq:P_k^(i), Wn FS,example}
\Pc_2^{(1)} & = \lef\{  W_{n, \Tc}^{(1)}: \type(\Tc)= \vv_2^{(1)}, n  \in [N]  \rig\}, \notag \\
\Pc_2^{(2)} & = \lef\{  W_{n, \Tc}^{(2)}: \type(\Tc)= \vv_2^{(2)}, n  \in [N]  \rig\}, \notag \\
\Pc_3^{(1)} & = \lef\{  W_{n, \Tc}^{(1),j}: j\in [2], \type(\Tc)= \vv_3^{(1)}, n  \in [N]  \rig \}.
\end{align}
\Aar, each file is split into 
$\Fpt=|\Pc_2^{(1)} | + |\Pc_2^{(2)} | + |\Pc_3^{(1)}|=2\binom{q+1}{1}\binom{q}{1}  +2\binom{q+1}{2}= \frac{3(K^2-1)}{4}  $ \pkts. 

With the above file  splitting, the cache placement and  \deli schemes are described  \af.

\tbf{Cache placement.} Cache placement is performed at the \sbf level. \Specly,
\user{k} stores  all \pkts associated with the \sbfs  $W_{n,\Tc}$ where $k\in \Tc$. Define $\Pc_{i, k}^{(j)} \eqdef \{w_\Tc \in  \Pc_{i}^{(j)}: k\in \Tc \}  $  as the subset of \tp-$\vv_i^{(j)}$ \pkts whose {support set} $\Tc$ covers  \user{k}. 
Then,
\be 
\label{eq:Zk,example}
Z_k =\Pc_{2, k}^{(1)}  \cup \Pc_{2, k}^{(2)} \cup \Pc_{3, k}^{(1)}, \; k\in [K] 
\ee
The \memoconst $|Z_k|=ML$ is satisfied and proved \af.  First, for any \user{k\in  \Qc_1}, we  have  $|\Pc_{2, k}^{(1)} | =|\Pc_{2, k}^{(2)}|=q$, and $|\Pc_{3, k}^{(1)}|=2q$. Thus, 
\begin{align}
|Z_k| & = \lef( |\Pc_{2, k}^{(1)}| + |\Pc_{2, k}^{(2)}|\rig)\ell^{(1)} + |\Pc_{3, k}^{(1)}|\ell^{(2)}\notag\\
 & \overset{(\ref{eq:pkt size,example})}{=}
  \lef( |\Pc_{2, k}^{(1)}| + \gamma|\Pc_{2, k}^{(2)}|+  |\Pc_{3, k}^{(1)}|   \rig)\ell^{(1)}\notag\\
& = \frac{(K^2-1)N\ell^{(1)} }{2} \overset{(\ref{eq:pkt size,example})}{=} \frac{2NL}{K}=ML,
\end{align}
where the last step is \bcuz $t=KM/N=2$. Second, for $ k\in  \Qc_2$,  $|\Pc_{2, k}^{(1)} | =|\Pc_{2, k}^{(2)}|=q+1$, and $|\Pc_{3, k}^{(1)}|=0$. Thus,
\begin{align}
|Z_k| & = |\Pc_{2, k}^{(1)}|\ell^{(1)}
+ |\Pc_{2, k}^{(2)}| \ell^{(2)} = \lef(|\Pc_{2, k}^{(1)}| + \gamma |\Pc_{2, k}^{(2)}|  \rig) \ell^{(1)} \notag\\
& =  \frac{(K^2-1)N\ell^{(1)} }{2} =ML.
\end{align}
\Aar, the \memoconst is satisfied.

\tbf{\Deli scheme.}
Let the user demands be $\dv=(d_1,\cdots,d_K)$. The \mtcst \deli phase consists of two {rounds}, each corresponding to the transmission of packet types within a distinct coupled group.
Note that the \deli within type-$\sv_1=(0,3)$ \mgrps is eliminated since \tp-$\vv_1$ \sbfs are excluded.

1) \tit{\Deli for \pkt \tps in $\Gc_1$:}
The \tx selection is given in (\ref{eq:tx select in G1}), with \invod \sbf \tps $\Ic_1=\{\vv_1^{(1)}\}  $, $\Ic_2=\{\vv_1^{(1)}, \vv_2^{(1)}\}  $, $\Ic_3=\{\vv_2^{(1)},\vv_3^{(1)} \}  $, and $\Ic_4=\{\vv_3^{(1)}\}$.
For a \tp-$\sv_2=(1^\dagger, 2)$ \mgrp
$\Sc= \{1, q+2, q+3\}$, \user{1} is the only \tx. It sends
$W_{d_{q+2}, \Sc \bksl \{q+2\} }^{(1)} \oplus W_{d_{q+3}, \Sc \bksl \{q+3\} }^{(1)} $, from which users  $q+2$ and  $q+3$ each  decodes a \tp-$\vv_2^{(1)}$ \pkt. 
For a \tp-$\sv_3=(2^\dagger,1)$ \mgrp $\Sc=\{1,2, q+2\}$, users $1$ and $2$ are the \txs. Let $(\pi_1, \pi_2 )$ be  a random permutation of $(1,2)$. Then \user{1} sends $ W_{d_2, \Sc \bksl  \{2\}   }^{(1)} \oplus  W_{d_{q+2}, \Sc \bksl  \{q+2\}   }^{(1), \pi_1}     $, and \user{2} sends $ W_{d_1, \Sc \bksl  \{1\}   }^{(1)} \oplus  W_{d_{q+2}, \Sc \bksl  \{q+2\}   }^{(1), \pi_2}     $. From this \deli,  users $1$ and $2$ each decodes a desired \tp-$\vv_2^{(1)}$ \pkt, while \user{q+2} decodes two \tp-$\vv_3^{(1)}$ \pkts.
Moreover, for a  \tp-$\sv_4=(3^\dagger,0)$ \mgrp $\Sc=\{1,2,3\}$, all users transmit, and the \cmsgs are given in Table~\ref{tab:deli for G1,example}, 
\begin{table}[ht]
\caption{\small \Deli within $\Sc=\{1,2,3\}$ under $\Gc_1$ }
\label{tab:deli for G1,example}
\vspace{-.05cm}
\centering
\setlength{\tabcolsep}{4pt}
\renewcommand{\arraystretch}{0.95}
\begin{tabular}{|c|c|c|}
\hline
\small Tx & \small\Cmsg & \small  Rx  \\
\thickhline 
$1$ &  $W_{d_2,\Sc \bksl \{2\}}^{(1),\pi_1^\prime} \oplus  W_{d_3,\Sc \bksl \{3\}}^{(1),\pi_1^{\prime\prime}}$        &  $\{2,3\}$ \\
\hline 
$2$ & $W_{d_1,\Sc \bksl \{1\}}^{(1),\pi_1} \oplus  W_{d_3,\Sc \bksl \{3\}}^{(1),\pi_2^{\prime\prime}}$     & $ \{1,3\}$ \\
\hline 
$3$ &  $W_{d_1,\Sc \bksl \{1\}}^{(1),\pi_2} \oplus  W_{d_2,\Sc \bksl \{2\}}^{(1),\pi_2^{\prime}}$     &  $\{1,2\}$ \\
\hline 
\end{tabular}
\end{table}
where $(\pi_1, \pi_2)$, $(\pi_1^\prime, \pi_2^\prime)$, and $(\pi_1^{\prime \prime}, \pi_2^{\prime \prime})$  are three random \perms of $(1,2)$.

2) \tit{\Deli for \pkt \tps in $\Gc_2$:}
The \tx selection is given in (\ref{eq:tx select in G2}), with \invod \sbf \tps $\Ic_1=\{\vv_1^{(2)}\}  $, $\Ic_2=\{\vv_1^{(2)}, \vv_2^{(2)}\}  $, $\Ic_3=\{\vv_2^{(2)},\vv_3^{(2)} \}  $, and $\Ic_4=\{\vv_3^{(2)}\}$. 
In $\Gc_2$, besides \tp $\vv_1$, the  \pkt type $\vv_3^{(2)}$ is also excluded since $\alphabm^{(2)}=(0,1,0)$. Hence, the \deli within \tp-$\sv_4$ \mgrps is eliminated.
The \deli within \tp-$\sv_2$ \mgrps is similar---as the \tx selection is the same---to that of $\Gc_1$ except that the \invod \pkts need to switch to types in $\Gc_2$.
For the \tp-$\sv_3$ \mgrp $\Sc=\{1,2, q+2\}$, \user{q+2} is the only \tx and sends $
W_{d_1, \Sc \bksl \{1\}}^{(2)} \oplus W_{d_2, \Sc \bksl \{2\}}^{(2)}  
$, from which each of users $1$  and $2$ decodes  a desired \tp-$\vv_2^{(2)}$ \pkt.

\tbf{Correctness.}
We show that each user can correctly recover the desired file. \Wlog, consider \user{1}. 
By (\ref{eq:Zk,example}), \user{1} has  stored $ | \Pc_{2,1}^{(1)}|/N=q$ \tp-$\vv_2^{(1)}$ \pkts,  $| \Pc_{2,1}^{(2)}|/N=q$ \tp-$\vv_2^{(2)}$ \pkts, and $ | \Pc_{3,1}^{(1)}|/N=2q$ \tp-$\vv_3^{(1)}$ \pkts of  the desired file $W_{d_1}$. 
We need to show that \user{1} can obtain the remaining \pkts of each \tp, \ie, $|\Pc_{2}^{(1)}\backslash \Pc_{2,1}^{(1)}|/N= q^2  $ \tp-$\vv_2^{(1)}$ \pkts, 
$| \Pc_{2}^{(2)}\bksl  \Pc_{2,1}^{(2)}|/N=q^2  $
\tp-$\vv_2^{(2)}$ \pkts, and $ | \Pc_{3}^{(1)}\bksl \Pc_{3,1}^{(1)} |/N=q(q-1)   $ \tp-$\vv_3^{(1)}$ \pkts.

We first show that \user{1} can decode all desired \pkts  of  \tps in $\Gc_1$.

\textbullet\;From each \tp-$\sv_3$ \mgrp  $\Sc= \{1,k,k^\prime \}  $ with $k \in \Qc_1\bksl\{1\}, k^\prime \in \Qc_2  $,  \user{1} decodes a  \tp-$\vv_2^{(1)}$ \pkt $W_{d_1, \Sc \bksl \{1\}}$. The number of such \mgrps is equal to $ \binom{|\Qc_1|-1}{1}\binom{|\Qc_2|}{1} =q^2$ so that \user{1} decodes a total of $q^2$ \tp-$\vv_2^{(1)}$ \pkts.

\textbullet\;From each \tp-$\sv_4$ \mgrp $\Sc=\{1,k,k^\prime \}  $ with $k,k^\prime \in \Qc_1 \bksl \{1\}$, \user{1} decode two \tp-$\vv_3^{(1)}$  \pkts  $W_{d_1, \Sc\bksl\{1\}}^{(1),1}$  and $W_{d_1, \Sc\bksl\{1\}}^{(1),2}$. There are $\binom{|\Qc_1|-1}{2}=\frac{q(q-1)}{2}$ such \mgrps. Hence, \user{1} decodes a total of $q(q-1)$ \tp-$\vv_3^{(1)}$  \pkts.

Second, consider \deli of \pkt \tps in $\Gc_2$. It can be easily seen that \user{1} decodes a total of $q^2$ \tp-$\vv_2^{(2)}$ \pkts. 
\Aar, \user{1} decodes all the desired \pkts, proving the correctness of the above \deli scheme. For users in $\Qc_2$, correctness can be proved similarly. \hfill $\lozenge$
\end{example}

\begin{remark}[Infeasibility of \Homo \Sbp]
\label{rmk: infeasibility of homo sbp}
\tit{For the user grouping $\qv=(q+1,q)$ in Example~\ref{example:t=2}, \homo \sbp  cannot satisfy the memory constraint (MC) while achieving any strict packet reduction over JCM.
For instance, neither \tx \selec in (\ref{eq:tx select in G1}) nor (\ref{eq:tx select in G2}) satisfies the MC as shown in (\ref{eq:MC not satisfied for each cpg,example}). 
In fact, no non-JCM FS vector \satisfs the MC.
\Specly, suppose a \gfsv $\alphabm=(\alpha_1,\alpha_2,\alpha_3)$ is obtained from  the $9$ \possb \tx selection \strtgs\footnote{There are a total of $9$ \possb \tx selection \strtgs since $\sv_2$ and  $\sv_3$ each admit $3$ \tx choices, while the selections for $\sv_1$ and $\sv_4$ are fixed.}. 
Since the MC requires $\alphabm \bm{\Delta}^{\rm T} =(\alpha_1,\alpha_2,\alpha_3)(q-1,1,-q)^{\rm T}
=(\alpha_1 - \alpha_3)q  + (\alpha_2 - \alpha_1)=0$ for any $q$, it must hold that $\alpha_1=\alpha_2=\alpha_3$. 
The only FS vector satisfying this  is $\alphabm_{\rm JCM
}=(2,2,2)$. \Aar, we cannot \achv \pkt \red with \homo \sbp while ensuring the MC.}
\end{remark}

\subsection{General \PTD }
\label{subsec: general pt design}

This section presents the general \PTD for \Thm \ref{thm:thm}, including file splitting,  cache placement, and \mtcst \deli phases.
We also analyze the resulting subpacketization relative to the JCM baseline and validate the packet-size ratio.

Let $(K,t)=(2q+1,2r)$ where $q \ge t+1$\footnote{The purpose of the assumption $q \ge t+1$ is to ensure the presence of all relevant \mgrp \tps. When $q<t+1$, the same \PTd works similarly.}.
Consider the \uneq \grpg $\qv=(q+1,q)$ with user assignment $\Qc_1=\{1, \cdots, q+1\}$ and $\Qc_2=\{q+2, \cdots, K\}$.
There are $V=t+1$ \sbf \tps and $S=t+2$ \mgrp \tps, 
\begin{subequations}
\label{eq:vk,sk def,gen pt}
\begin{align}
\vv_k&=(k-1,t-k+1),\;  k\in[t+1]\\
\sv_k&=(k-1,t-k+2),\;  k\in[t+2]
\end{align}
\end{subequations}
with \invod \sbf type sets $\Ic_1=\{\vv_1\}$, $\Ic_k=\{\vv_{k-1}, \vv_k\}, k \in [2:t+1]$, and $\Ic_{t+2}=\{\vv_{t+1}\}$. 
Each \sbf type is divided into two \tit{\sbtps} $\vv_k=(\vv_k^{(1)}, \vv_k^{(2)})$ \soth the \cpgrps are 
\be
\label{eq:Gi def, gen PT design}
\Gc_i= \big\{ \vv_k^{(i)}\big\}_{k=1}^{t+1}, \; i=1,2
\ee 
All \pkts of \sbtps in $\Gc_i$ have the same size of $\ell^{(i)}$ bits.

\tit{Determination of Global FS Vectors.}
Each \cpgrp corresponds to a  specific \tx \selec \af.

For $\Gc_1$, the \tx \selec is 
\begin{align}
\label{eq:tx select G1, gen PT}
\sv_k
=\left\{\begin{array}{ll}
\left(0,(t+1)^{\dagger}\right), & \text{if}\; k=1\\
\left((k-1)^{\dagger},t-k+2\right), & \text{if}\;k\in[2:t+1]\\
\left((t+1)^{\dagger},0\right), &\text{if}\; k=t+2
\end{array}\right.
\end{align} 
The \lfsvs are\footnote{For ease of notation,  the $\star$ symbols  are omitted in each $\alphabm_k$.} $ \alphabm_1=\big(\alpha(\vv_1^{(1)}\big)=(t)$,
$\alphabm_k=\big(\alpha(\vv_{k-1}^{(1)}), \alpha(\vv_{k}^{(1)}) \big)=(k-2, k-1), k\in [2:t+1]$, and $\alphabm_{t+2}=\big(\alpha(\vv_{t+1}^{(1)})\big)=(t)  $, yielding a \gfsv for $\Gc^{(1)}$ as
\be 
\label{eq:alpha^(1), gen PT}
\bm{\alpha}^{(1)}=(0,1,2,\cdots,t-1,t),
\ee 
implying that subtype $\vv_1^{(1)}$ is excluded. 
For $\Gc_2$, the \tx \selec is 
\begin{align}
\label{eq:tx select G2, gen PT}
\sv_k
=\left\{\begin{array}{ll}
\left(0,(t+1)^{\dagger}\right), & \text{if}\; k=1\\
\left((k-1)^{\dagger},t-k+2\right), & \text{if}\;k\in[2:r+1]\\
\left(k-1,(t-k+2)^{\dagger}\right), &\text{if} \; k\in[r+2:t+1]\\
\left((t+1)^{\dagger},0\right), &\text{if} \; k=t+2
\end{array}\right.
\end{align}
The \corrspdg \lfsvs are $\bm{\alpha}_1=(t)$, 
$\bm{\alpha}_k=(k-2,k-1), k\in[2:r+1]$,    
$\bm{\alpha}_k=\big(t-k+2,t-k+1\big), k\in[r+2:t+1]$,  and $\bm{\alpha}_{t+2}=(t)$, yielding \gfsv for $\Gc^{(2)}$ as
\be 
\label{eq:alpha^(2), gen PT}
\bm{\alpha}^{(2)} = (0,1,2,\cdots,r-1,r,r-1,\cdots,2,1,0),
\ee
implying the exclusion of subtypes $\vv_1^{(2)}$ and $\vv_{t+1}^{(2)}$. 

Adding (\ref{eq:alpha^(1), gen PT}) and (\ref{eq:alpha^(2), gen PT}), the \gfsv becomes 
\be 
\label{eq:alpha^global, gen PT}
\alphaglobal = \bm{\alpha}^{(1)} + \bm{\alpha}^{(2)}
= \lef(0,2,4,\cdots,t-2, \underline{t}_{r+1} \rig),
\ee
where \tp-$\vv_1$ \sbfs/\pkts are excluded.

Fig.~\ref{fig:alphas comparison,gen pt} shows  the global FS vectors for $t=8$. 
The black and blue curves are two intermediate FS vectors, whose element-wise sum yields the red curve, i.e., the final global FS vector that determines the \sbp level. The red curve has a distinctive shape: in the first half, it is strictly below $t$ (the JCM splitting factor), so the \sbp reduction comes from assigning fewer than $ t$ packets to these subfile types; in the second half, it is capped at $t$, due to the opposite monotonic trends of the two intermediate vectors, whose sum saturates at $t$.
\begin{figure}
    \centering
    \includegraphics[width=0.34\textwidth, height =0.27\textwidth]{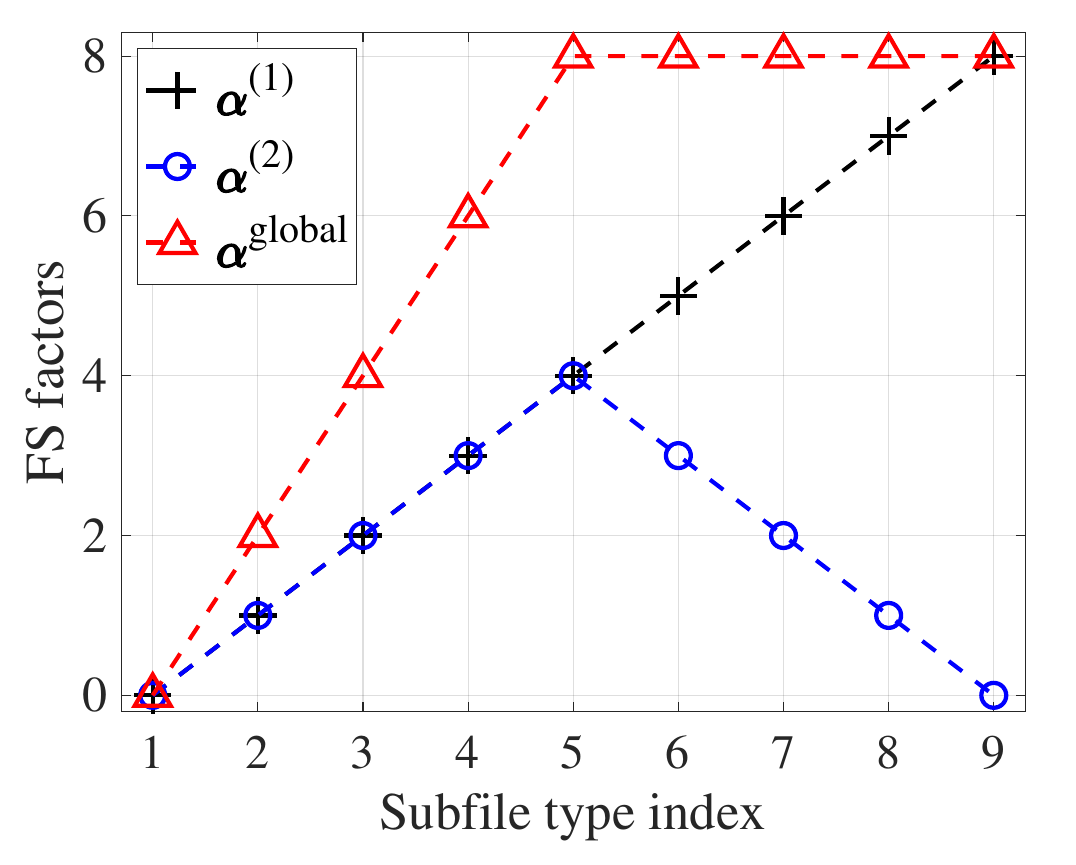}
    \vspace{-.2cm}
    \caption{\small Comparison of global FS vectors  for $t=8$.}
    \label{fig:alphas comparison,gen pt}
\end{figure}
The resulting  \sbp is equal to 
\be
\label{eq:Fpt,gen pt}
\Fpt= \sum\nolimits_{k=1}^{t+1}\alphaglobal(k)F(\vv_k),
\ee
as shown in \Thm \ref{thm:thm},
where $F(\vv_k) = \binom{q+1}{k-1}\binom{q}{t-k+1}$  is the number of \tp-$\vv_k$ \sbfs, $k\in [t+1]$.

\tit{Calculation of \Pkt Sizes.}
The \pkt sizes $\ell^{(1)}$ and $\ell^{(2)}$ are determined by the \memoconst  of  the users, \ie, each user in $\Qc_1$ and $\Qc_2$ should store the same number of $ML$ bits. \Specly,
\begin{subequations}
\label{eq:MC,gen pt}
\begin{align}
\lef( \bm{\alpha}^{(1)}\mathbf{F}_1^{\rm T}\ell^{(1)}+
\bm{\alpha}^{(2)}\mathbf{F}_1^{\rm T}\ell^{(2)}
\rig)N & = ML,
\label{eq:MC1,gen pt}\\
\lef( \bm{\alpha}^{(1)}\mathbf{F}_2^{\rm T}\ell^{(1)}+
\bm{\alpha}^{(2)}\mathbf{F}_2^{\rm T}\ell^{(2)}
\rig)N & = ML,
\label{eq:MC2,gen pt}
\end{align}
\end{subequations}
where 
\be
\label{eq:L,gen pt}
L= \alphabm^{(1)}\Fm^{\rm T}\ell^{(1)} +
\alphabm^{(2)}\Fm^{\rm T}\ell^{(2)}.
\ee 
The  \sbf number and user cache vectors are 
\begin{align}
\Fm & =[F(\vv_k)]_{k=1}^{t+1},\;  F(\vv_k) = \binom{q+1}{k-1}\binom{q}{t-k+1},\label{eq:F,gen pt} \\
\Fm_1 & =[F_1(\vv_k)]_{k=1}^{t+1},\;  F_1(\vv_k) = \binom{q}{k-2}\binom{q}{t-k+1},\label{eq:F1,gen pt} \\
\Fm_2 & =[F_2(\vv_k)]_{k=1}^{t+1},\;  F_2(\vv_k) = \binom{q+1}{k-1}\binom{q-1}{t-k}.\label{eq:F2,gen pt}
\end{align}
Subtracting (\ref{eq:MC1,gen pt})  from (\ref{eq:MC2,gen pt}),  we get  
\be 
\label{eq:differential MC, gen pt}
\big(\bm{\alpha}^{(1)}\ell^{(1)}
+ \bm{\alpha}^{(2)}\ell^{(2)}\big)\Deltam^{\rm T}=0,
\ee
yielding  the packet  size ratio 
\be 
\label{eq:gamma, gen pt}
\gamma \eqdef  \frac{\ell^{(2)}}{\ell^{(1)}} 
= -\frac{\bm{\alpha}^{(1)}\Deltam^{\rm T}     }
{\bm{\alpha}^{(2)}\Deltam^{\rm T}  }, 
\ee
where  $ \Deltam=\Fm_2- \Fm_1$ is the  user cache \diffce vector. 
Substituting (\ref{eq:gamma, gen pt}) in (\ref{eq:L,gen pt}), the packet sizes are obtained as
\be 
\label{eq:pkt size, gen pt}
\ell^{(1)} = \frac{L}{(\bm{\alpha}^{(1)} + \gamma \bm{\alpha}^{(2)} )\mathbf{F}^{\rm T}},
\;\; 
\ell^{(2)}= \gamma  \ell^{(1)}.
\ee

\tbf{File splitting.}
Under the above \PTd, each file is split into $\Fpt$ (see (\ref{eq:Fpt,gen pt})) \pkts, including $\alphabm^{(1)}\Fm^{\rm T}$ \pkts of size  $\ell^{(1)}$ and $\alphabm^{(2)}\Fm^{\rm T}$ \pkts of size  $\ell^{(2)}$ bits. \Ip, 
\begin{align}
\label{eq:Wn splitting,gen pt}
W_n= \lef( \bigcup\nolimits_{k\in [2:t+1]} \Pc_k^{(1)}    \rig) \cup \lef( \bigcup\nolimits_{k \in [2:t]  }\Pc_k^{(2)}    \rig),
\end{align}
where
\be 
\Pc_k^{(i)} \eqdef  \lef\{ W_{n, 
 \Tc}^{(i),j}:j \in \big[\alphabm^{(i)}(k)\big], \type(\Tc)= \vv_k^{(i)}, n\in  [N]      \rig\}
\ee 
denotes the set of \tp-$\vv_k^{(i)}$ \pkts and $|\Pc_k^{(i)}|=N \alphabm^{(i)}(k)F(\vv_k),k\in [t+1], i \in[2]$.

The cache  placement and \mtcst \deli phases are described  \af.

\tbf{Cache placement.}
\User{u\in [K]} stores all \pkts $W_{n, \Tc}^{(i),j}$
where $ u \in \Tc$. 
Let $\Pc_{k, u}^{(i)} \eqdef  \big\{w_\Tc \in  \Pc_{k}^{(i)}: u \in \Tc \big\}$ denote  the subset of \tp-$\vv_k^{(i)}$ \pkts whose support set $\Tc$ includes $k$, where $|\Pc_{k, u}^{(i)}| = \alphabm^{(i)}(k)F_j(\vv_k)  $ if $u \in \Qc_j, j=1,2$. Then, 
\be 
\label{eq:Zk,gen pt}
Z_u= 
\lef( \bigcup\nolimits_{k\in [2:t+1]} \Pc_{k,u}^{(1)}    \rig) \cup \lef( \bigcup\nolimits_{k \in [2:t]  }\Pc_{k,u}^{(2)}    \rig), u\in[K]
\ee
It can be verified that $|Z_u|=ML,\forall u\in[K]$. The proof is given in Appendix~\ref{sec:proof of Zu MC, appendix}.

\tbf{\Deli \schm.}
Let the users' demands be $\dv=(d_1,\cdots, d_K)$. 
The \mtcst \deli phase consists of two rounds, each corresponding to one coupled group. The first round follows the \tx \selec in (\ref{eq:tx select G1, gen PT}) and \invos only \pkt \tps in $\Gc_1$.  
The second round follows (\ref{eq:tx select G2, gen PT}) and \invos \pkt \tps from $\Gc_2$.

To simplify exposition, note  that 
the transmitter selections in (\ref{eq:tx select G1, gen PT}) and (\ref{eq:tx select G2, gen PT}) coincide for  
types $\sv_2, \cdots, \sv_{r+1}$ and $\sv_{t+2}$ in both $\Gc_1$ and $\Gc_2$, and differ only for types $\sv_{r+2}, \cdots, \sv_{t+1}$. 
Since $\sv_1$ \invos only excluded \tp-$\vv_1$  \sbfs, \tp-$\sv_1$ \mgrps are omitted from \deli.

1) \tit{\Deli in $\sv_2, \cdots, \sv_{r+1}$ and $\sv_{t+2}$:}  
\Accrdto (\ref{eq:alpha^(2), gen PT}), \tp-$\vv_{t+1}^{(2)}$ \pkts are excluded. Thus,  the \deli for \tp-$\sv_{t+2}$ (note that $ \Ic_{t+2}=\{\vv_{t+1}\}$) under $\Gc_2$ is omitted.

\textbullet\tit{ \Deli in $\sv_2, \cdots, \sv_{r+1}$.}
    Consider a \tp-$\sv_k=( (k-1)^\dagger, t-k+2  )$ \mgrp $\Sc$ where $|\Sc\cap \Qc_1|=k-1$,
    $|\Sc\cap \Qc_2|=t-k+2$. The users in $\Sc \cap \Qc_1$ are selected as \txs. Under $\Gc_i, i=1,2$, each  \user{u\in \Sc \cap \Qc_1} sends
    \be
    \label{eq:cmsg, sv_2,...,sv_r+1,gen pt}
    \Bigg( \bigoplus_{u^\prime  \in \Sc \cap \Qc_1 \bksl \{u \}   }    W_{d_{u^\prime }, \Sc  \bksl \{u^\prime \}}^{(i), \bm{\pi}_{u^\prime}(u^\prime )} \Bigg) 
    \oplus 
    \Bigg( \bigoplus_{v\in \Sc \cap \Qc_2}    W_{d_{v}, \Sc  \bksl \{v\}}^{(i), \bm{\pi}(v)} \Bigg)  ,
    \ee 
    where, for each $u \in \Sc \cap \Qc_1$,     $\bm{\pi}_u$ is a random bijection from $\Sc \cap \Qc_1 \bksl \{u\}$ to $(1, \cdots,k-2)$, and  $\bm{\pi} $ is a random bijection from $\Sc \cap \Qc_1$ to $(1, \cdots,k-1)$.
    The random bijections coordinate packet assignments across multicast messages, ensuring each packet is delivered exactly once without redundancy.
    From (\ref{eq:cmsg, sv_2,...,sv_r+1,gen pt}), each user in $\Sc\cap \Qc_1$ decodes $k-2$ \tp-$\vv_{k-1}^{(i)}$ \pkts, whereas each user in $\Sc\cap \Qc_2$ decodes $k-1$ \tp-$\vv_{k}^{(i)}$ \pkts.

\textbullet\tit{ \Deli in $\sv_{t+2}$.}  Under $\Gc_1$, all users in each \tp-$\sv_{t+2}=((t+1)^\dagger,0)$ \mgrp $\Sc$ are chosen as \txs. Each \user{u\in \Sc} sends
\be 
\label{eq:cmsg for s_t+2,gen pt}
\bigoplus\nolimits_{u^\prime \in \Sc \bksl  \{u\} } W_{d_{u^\prime}, \Sc \bksl \{u^\prime\}  }^{(1), \bm{\pi}_{u^\prime}(u^\prime)   },
\ee 
where, for each $u\in \Sc$, $\bm{\pi}_{u}$ is a random bijection from $\Sc$ to $(1,\cdots,t)$. 
From (\ref{eq:cmsg for s_t+2,gen pt}), each user decodes $t$ \tp-$\vv_{t+1}^{(1)}$ \pkts.

2) \tit{\Deli in $\sv_{r+2}, \cdots, \sv_{t+1}$:}
Consider a \tp-$\sv_k=(k-1, t-k+2)$ \mgrp 
$\Sc$ where  $|\Sc \cap \Qc_1|=k-1$, $|\Sc \cap \Qc_2|=t-k+2$ and \invod \sbf types $\Ic_k=\{\vv_{k-1},  \vv_k\}$ for some $k\in [r+2:t+1]$.

\textbullet\tit{ \Deli for $\Gc_1$.} The \tx \selec is $\sv_k=((k-1)^\dagger, t-k+2)$, implying that the users in $\Sc\cap \Qc_1$ are selected as \txs. 
\Ip, each \user{u \in \Sc\cap \Qc_1} sends
\be
\label{eq:cmsg sk G1,deli schm,gen pt}
\Bigg( \bigoplus_{u^\prime  \in \Sc \cap \Qc_1 \bksl \{u \}   }    W_{d_{u^\prime }, \Sc  \bksl \{u^\prime \}}^{(1), \bm{\pi}_{u^\prime}(u^\prime )} \Bigg) 
\oplus 
\Bigg( \bigoplus_{v\in \Sc \cap \Qc_2}    W_{d_{v}, \Sc  \bksl \{v\}}^{(1), \bm{\pi}(v)} \Bigg)  ,
\ee 
where the random bijections $\{\bm{\pi}_u\}_{u \in\Sc \cap \Qc_1}$ and $\bm{\pi}$ are  defined similarly to that in
(\ref{eq:cmsg, sv_2,...,sv_r+1,gen pt}).
From (\ref{eq:cmsg sk G1,deli schm,gen pt}), each user in $\Sc \cap  \Qc_1$ decodes $k-2$ \tp-$\vv_{k-1}^{(1)}$ \pkts, whereas each user in $\Sc \cap  \Qc_2$ decodes $k-1$ \tp-$\vv_{k}^{(1)}$ \pkts.

\textbullet\tit{ \Deli for $\Gc_2$.} The \tx \selec is $\sv_k=(k-1, (t-k+2)^\dagger)$ \soth the \txs are users in $\Sc\cap \Qc_2$. Each \user{v\in \Sc \cap \Qc_2} sends
\be
\label{eq:cmsg sk G2,deli schm,gen pt}
\Bigg( \bigoplus_{u \in \Sc \cap \Qc_1  }    W_{d_u, \Sc  \bksl \{u\}}^{(2), \bm{\pi}(u)}  \Bigg) 
\oplus
\Bigg( \bigoplus_{v^\prime \in \Sc \cap \Qc_2 \bksl \{v\}   }    W_{d_{v^\prime}, \Sc  \bksl \{v^\prime\}}^{(2), \bm{\pi}_{v^\prime}(v^\prime)}  \Bigg), 
\ee
where, for each $v \in \Sc\cap \Qc_2$, $\bm{\pi}_v$ is a random bijection from $\Sc\cap \Qc_2 \bksl\{ v\}$ to $(1, \cdots, t-k+1)$, and  $\bm{\pi} $ is a random bijection from $\Sc \cap \Qc_2$ to $(1, \cdots,t-k+2)$. 
From (\ref{eq:cmsg sk G2,deli schm,gen pt}), each user in $\Sc\cap  \Qc_1 $ decodes $t-k+2$ \tp-$\vv_{k-1}^{(2)}$ \pkts, whereas each user in n $\Sc\cap  \Qc_2 $ decodes $t-k+1$ \tp-$\vv_k^{(2)}$ \pkts.

\tbf{Proof of correctness.}
\Accrdto (\ref{eq:Wn splitting,gen pt}) and (\ref{eq:Zk,gen pt}), each \user{u \in \Qc_j,j=1,2} needs to recover the following \pkts from the desired  file $W_{d_u}$:
$| \Pc_{k}^{(1)} \bksl  \Pc_{k,u}^{(1)}|/N=  \alphabm^{(1)}(k)(F(\vv_k) - F_j(\vv_k)) $ \tp-$\vv_k^{(1)}$ \pkts  for all  $k \in [2:t+1]$, and  
$| \Pc_{k}^{(2)} \bksl  \Pc_{k,u}^{(2)}|/N=  \alphabm^{(2)}(k)(F(\vv_k) - F_j(\vv_k)) $ \tp-$\vv_k^{(2)}$ \pkts  for all  $k \in [2:t]$. 

Next we prove that \user{u} can indeed recover all the desired \pkts. \Wlog, suppose $u\in \Qc_1$.

\textbullet\tbf{\;Recovery from  $\sv_2, \cdots, \sv_{r+1}$.}
From each \tp-$\sv_k,k\in [2:r+1]  $ \mgrp $\Sc \ni u$, \user{u} decodes $k-2$ \tp-$\vv_{k-1}^{(1)}$
and \tp-$\vv_{k-1}^{(2)}$
\pkts, \resp, 
\accrdto (\ref{eq:cmsg, sv_2,...,sv_r+1,gen pt}). There are $\binom{|\Qc_1|-1}{k-2}\binom{|\Qc_2|}{t-k+2}=\binom{q}{k-2}\binom{q}{t-k+2}  $ such \grps, so \user{u} decodes in total   $(k-2)\binom{q}{k-2}\binom{q}{t-k+2} $ \tp-$\vv_{k-1}^{(i)}$ \pkts for each $i=1,2$.
On the  other hand, 
${| \Pc_{k-1}^{(i)} \bksl  \Pc_{k-1,u}^{(i)}|}/{N} = (k-2)\big( \binom{q+1}{k-2}\binom{q}{t-k+2} -\binom{q}{k-3}\binom{q}{t-k+2}    \big)= (k-2)\binom{q}{k-2}\binom{q}{t-k+2},i=1,2 $. Hence,
the number of recovered \pkts matches the number of desired \tp-$\vv_{k-1}^{(i)}$ \pkts.

\textbullet\tbf{\;Recovery from  $ \sv_{t+2}$.} 
From each \tp-$\sv_{t+2}$ \mgrp $\Sc \ni u$, \user{u} decodes $t$ \tp-$\vv_{t+1}^{(1)}$ \pkts \accrdto (\ref{eq:cmsg for s_t+2,gen pt}). Since there are $ \binom{|\Qc_1|-1}{|\Sc|-1}=\binom{q}{t}$ such \grps, \user{u} decodes in total   $t \binom{q}{t}  $ \tp-$\vv_{t+1}^{(1)}$ \pkts.
On the  other hand, 
 ${| \Pc_{t+1}^{(1)} \bksl  \Pc_{t+1,u}^{(1)}|}/{N} =t \big( \binom{q+1}{t}\binom{q}{0}  - \binom{q}{t-1}\binom{q}{0} \big)=t\binom{q}{t}$. Hence,
 the number of recovered \pkts matches the number of desired \tp-$\vv_{t+1}^{(1)}$ \pkts. All such \pkts are thus recovered.

\textbullet\tbf{\;Recovery from $ \sv_{r+2}, \cdots, \sv_{t+1}$.}
For each $k \in [r+2:t+1]$, there are $\binom{q}{k-2}\binom{q}{t-k+2}  $ \tp-$\sv_k $ \mgrps containing  \user{u}, from each of which \user{u} decodes  $k-2$ 
\tp-$\vv_{k-1}^{(1)}$ \pkts (see (\ref{eq:cmsg sk G1,deli schm,gen pt})) and  $t-k+2$ \tp-$\vv_{k-1}^{(2)}$ \pkts (see (\ref{eq:cmsg sk G2,deli schm,gen pt})).
Thus, in total $(k-2)\binom{q}{k-2}\binom{q}{t-k+2}  $ \tp-$\vv_{k-1}^{(1)}$ and $(t-k+2)\binom{q}{k-2}\binom{q}{t-k+2}  $ \tp-$\vv_{k-1}^{(2)}$ \pkts are decoded. 
\Movr, 
$|\Pc_{k-1}^{(1)} \bksl \Pc_{k-1,u}^{(1)}|/N=
(k-2)\binom{q}{k-2}\binom{q}{t-k+2}$ and 
$|\Pc_{k-1}^{(2)} \bksl \Pc_{k-1,u}^{(2)}|/N=
(t-k+2)\binom{q}{k-2}\binom{q}{t-k+2}$.
Hence, the number of recovered \pkts matches the desired \pkts. 

\Aar, we have proved that each  user in $\Qc_1$ can recover all the missing \pkts of  the desired file. The proof  for $\Qc_2$ is similar. Since each \cmsg is  useful to $t$ other users, the \comm rate $\Rjcm=N/M-1$ is \achvd.

\subsection{\Sbp Analysis}
\label{subsec:subpkt analysis, gen pt}

The proposed \PTd \achvs \sbp  
$\Fpt= \sum_{k=1}^{t+1} \alpha_k f_k$, where
$\alpha_k \eqdef \alphaglobal(k)$ (see (\ref{eq:alpha^global, gen PT})) and $f_k \eqdef F(\vv_k)$ (see (\ref{eq:F,gen pt}))
denote  the \gfsf  and number of \tp-$\vv_k$ \sbfs, \resp, for $k\in [t+1]$.
The JCM \sbp can be written as (see \Rmk \ref{rmk:JCM sbp identity})
\be
\label{eq:JCM sbp, sbp analysis,gen pt}
\Fjcm= \sum\nolimits_{k=1}^{t+1} \alphabm_{\rm JCM}(k)F(\vv_k),
\ee 
where  $\alphabm_{\rm JCM}  \eqdef (t, \cdots,t)$ is the JCM FS vector. 
Comparing (\ref{eq:Fpt,gen pt}) and (\ref{eq:JCM sbp, sbp analysis,gen pt})  suggests  $\Fpt \le \Fjcm$ since $ \alphaglobal \preceq \alphabm_{\rm JCM}      $.  \Thf, both the \ssgain and \fssgain are attainable: the former arises from excluding \tp-$\vv_1$ \sbfs, while the latter results from smaller FS factors than the \jsch.
As shown in Fig.~\ref{fig:alphas comparison,gen pt}, the  \sbp ratio $ \Fpt/\Fjcm$ \decrss in $q$, which is formally proved in the following lemma.

\begin{lemma}
\label{lemma:monotonicity of Fpt/Fjcm}
\tit{Define $\rho_t(q)\eqdef  \frac{F_{\rm PT}(q,t)}{F_{\rm JCM}(q,t)}   $ as the \sbp ratio of PT (cf.  (\ref{eq:Fpt, thm})) over  JCM (cf. (\ref{eq:Fjcm,thm implication})) for given $t$ and $q$. Then, for any fixed $t$, $\rho_t(q)$ is strictly decreasing in $q$.
}
\end{lemma}
\begin{IEEEproof}
See \App \ref{sec:proof of ratio monocity,app}.
\end{IEEEproof}

\begin{remark}[Subfile Count Identity]
\label{rmk:JCM sbp identity} 
\tit{The identity $\sum_{k=1}^{t+1}F(\vv_k) = \binom{K}{t}$ holds for any $K$ and $t$. To prove this, define $\Vc_k \eqdef  \{\Tc\subseteq [K]  : \type(\Tc)=\vv_k  \}$ as the set of  \tp-$\vv_k$ $t$-\sbsts of $[K]$. By definition, $|\Vc_k|=F(\vv_k),\forall k \in[t+1]$.
First, every $t$-\sbst  $\Tc$ has a unique  \tp $\vv_k$ for some  $k \in [t+1]$. Hence, $\binom{K}{t}\le \sum_{k=1}^{t+1}F(\vv_k)$.
Conversely,
\sbsts of \diff \tps are disjoint,  \ie, $  \Vc_i \cap \Vc_j=\emptyset, \forall i\ne  j$. \Thf, $ \sum_{k=1}^{t+1}F(\vv_k) \le \binom{K}{t}$. 
Combining the two inequalities yields $\sum_{k=1}^{t+1}F(\vv_k)= \binom{K}{t}$.
\Aar, $\Fjcm$ can be written as $\Fjcm=t \sum_{k=1}^{t+1}F(\vv_k)$.}
\end{remark}

Next, we  derive the \asympt bound (\ref{eq:asymp bound,thm}) in \Thm \ref{thm:thm}.
Note  that $\Fpt$ (see (\ref{eq:Fpt,gen pt}))  can be written as 
$ \Fpt = \sum_{k=1}^{r}2(k-1)F(\vv_k) + t\big(\sum_{k=r+1}^{t+1}F(\vv_k)\big)$. 
Thus,
\begin{align}
\label{eq:Fjcm-Fpt,sbp analysis,gen pt}
\Fjcm -\Fpt & \overset{(\ref{eq:JCM sbp, sbp analysis,gen pt})}{=} t \lef( \sum\nolimits_{k=1}^{t+1}F(\vv_k)\rig) -  \Fpt \notag\\
& =\sum\nolimits_{k=1}^{t+1}\lef(t-2(k-1)\rig)  F(\vv_k),
\end{align}
where $F(\vv_k) \overset{(\ref{eq:F,gen pt})}{=}\binom{q+1}{k-1}\binom{q}{t-k+1}=\frac{\prod_{i=0}^{k-2}(q+1-i)\prod_{i=0}^{t-k}(q-i) } {(k-1)!(t-k+1)!}$,  $k \in [t+1]$.
Fix $t$.
For sufficiently large $q$,
$F(\vv_k) \approx \frac{q^t}{(k-1)!(t-k+1)!}$.  
\Thf, 
\begin{subequations}
\label{eq:lim (1-Fpt/Fjcm)}
\begin{align}
& \lim_{q\to\infty}\left(1-   \frac{F_{\rm PT}}{F_{\rm JCM}}\right)
\overset{(\ref{eq:Fjcm-Fpt,sbp analysis,gen pt})}{=}
\lim\limits_{q\to\infty}\frac{1}{t}\frac{\sum_{k=1}^{r} \frac{\left(t-2(k-1)\right)q^t}{(k-1)!\left(t-(k-1)\right)!}     }{\sum_{k=1}^{t+1}\frac{q^t}{(k-1)!\left(t-(k-1) \right)!}} \notag \\
 & = 
\frac{1}{t}\frac{\sum_{k=1}^{r} \frac{t-2(k-1)}{(k-1)!\left(t-(k-1)\right)!}  }{\sum_{k=1}^{t+1}\frac{1}{(k-1)!\left(t-(k-1) \right)!}}=
\frac{1}{t}\frac{\sum_{k=0}^{r-1} \frac{t-2k}{k!\left(t-k\right)!}     }{\sum_{k=0}^{t}\frac{1}{k!\left(t-k\right)!}} \\
 &=
\frac{1}{t}\frac{\sum_{k=0}^{r-1} \frac{(t-2k)t!}{k!\left(t-k\right)!}     }{\sum_{k=0}^{t}\frac{t!}{k!\left(t-k\right)!}}=
\frac{1}{t}\frac{\sum_{k=0}^{r-1}(t-2k)\binom{t}{k}       }{\sum_{k=0}^{t}\binom{t}{k}} \overset{\trm{(a)}}{=} \frac{\binom{t-1}{r-1}}{2^t}\\
& \overset{\trm{(b)}}{=}  \frac{\frac{r}{t}\binom{t}{r}}{2^t}
=\frac{1}{2}\frac{\binom{t}{r}}{2^t} 
\Rightarrow 
\lim_{q\to\infty}\frac{F_{\rm PT}}{F_{\rm JCM}}= 1 - \frac{1}{2}\cdot \frac{\binom{t}{r}}{2^t}, 
\end{align}
\end{subequations}
where in step (b) we used the identity 
$ \binom{t-1}{r-1}=\frac{r}{t}\binom{t}{r}   $; (a) is \bcuz 
$ \sum_{k=0}^{r-1}(t-2k)\binom{t}{k} =t\binom{t-1}{r-1}$, which is proved \af.

\tit{Proof of step (a):}
Let $S\eqdef  \sum_{k=0}^{r-1}(t-2k)\binom{t}{k}= t\sum_{k=0}^{r-1}\binom{t}{k}-2\sum_{k=0}^{r-1}k\binom{t}{k} $. Using  the identity $k\binom{t}{k}=t\binom{t-1}{k-1}$, we have 
$ \sum_{k=0}^{r-1}k\binom{t}{k}=t\sum_{k=1}^{r-1}\binom{t-1}{k-1}= t\sum_{j=0}^{r-2}\binom{t-1}{j}$. Hence,
\be
\label{eq:S}
S = t\lef( \sum\nolimits_{k=0}^{r-1}\binom{t}{k} - 2 \sum\nolimits_{j=0}^{r-2}\binom{t-1}{j}  \rig).
\ee 
Applying the identity $\binom{t}{k}= \binom{t-1}{k} + \binom{t-1}{k-1}$,  
$\sum_{k=0}^{r-1}\binom{t}{k}
=
\sum_{k=0}^{r-1}\binom{t-1}{k} + \sum_{k=0}^{r-1}\binom{t-1}{k-1} 
=
\sum_{k=0}^{r-1}\binom{t-1}{k} + \sum_{j=0}^{r-2}\binom{t-1}{j}
$. 
Plugging this into (\ref{eq:S}), we obtain
\be
 S = t\lef( \sum_{k=0}^{r-1}\binom{t-1}{k} -  \sum_{j=0}^{r-2}\binom{t-1}{j}  \rig)=t\binom{t-1}{k-1}.
\ee

\tit{\Asympt scaling:}
Using Stirling's formula $n! =\sqrt{2\pi n}(n/e)^n(1+o(1))$~\cite{robbins1955remark}, we can easily obtain $ \binom{t}{r}/2^t \approx \sqrt{ {2}/{(\pi t)}}$:
\begin{align}
\frac{\binom{t}{r}}{2^t} & = \frac{(2r)!}{4^r(r!)^2 } 
= 
\frac{\sqrt{4\pi r}(2r/e)^{2r}(1+o(1)) } {4^r\cdot2\pi r(r/e)^{2r}(1+o(1))}  \approx \sqrt{\frac{2}{\pi t}}.
\end{align}
Hence, $ \lim_{q\to\infty} \Fpt/\Fjcm \approx  1- \sqrt{1/(2\pi t)}= 1-\Theta(1/\sqrt{t})$, proving the \asympt scaling in \Thm \ref{thm:thm}.

\subsection{Validation of Packet Size Ratio}
\label{subsec:pkt size ratio validation,gen pt}

We prove that the \pkt size ratio $ \gamma=\ell^{(2)}/\ell^{(1)}$ obtained in (\ref{eq:gamma, gen pt}) is valid, \ie, $\gamma$ is a positive  rational number. 
Since both the numerator and denominator are integers,  $\gamma$ is rational. We only need to prove $\gamma>0$.

\begin{lemma}
\label{lemma:valid gamma,gen pt}
\textit{For the general PT design 
in Section~\ref{subsec: general pt design}, 
the packet size ratio (\ref{eq:gamma, gen pt}) is valid, i.e., 
$\gamma=-\frac{\bm{\alpha}^{(1)}\bm{\Delta}^{\rm T}      }{\bm{\alpha}^{(2)}\bm{\Delta}^{\rm T}} >0$.
}\end{lemma}
\begin{IEEEproof} 
We prove $\gamma >0$ by showing 
$\bm{\alpha}^{(1)}\bm{\Delta}^{\rm T}<0$ and
$\bm{\alpha}^{(2)}\bm{\Delta}^{\rm T}>0$, \resp. 
Claims~\ref{claim:1} and \ref{claim:2} (see below) are used  to prove $ \bm{\alpha}^{(1)}\bm{\Delta}_1^{\rm T}<0$, while Claim~\ref{claim:3} is needed to prove $\bm{\alpha}^{(2)}\bm{\Delta}_1^{\rm T}>0$.
The proofs of the claims are provided in Appendix~\ref{appendix: proof of calims}.
\end{IEEEproof}

\begin{claim}
\label{claim:1}
\textit{The user cache difference vector $\bm{\Delta} \eqdef \Fm_2 -\Fm_1$ (see (\ref{eq:F1,gen pt}), (\ref{eq:F2,gen pt})) satisfies $\Deltam(k)>0, \forall k \in[r+1]$ and
$\Deltam(k)<0, \forall k\in [r+2:t+1]$.
}\end{claim}
\begin{IEEEproof}
See Appendix~\ref{appendix: proof of claim 1}. By Claim~\ref{claim:1}, the first
$r+1$ entries of the  \seq  $(\Deltam(k))_{k=1}^{t+1}$ are positive, whereas the last $r$ entries are negative.
\end{IEEEproof}

\begin{claim}
\label{claim:2}
\textit{The sum of the entries of $\bm{\Delta}$ is equal to zero, \ie,
$\sum_{i=1}^{t+1}\Deltam(i)=0$. 
}\end{claim}
\begin{IEEEproof}
See Appendix~\ref{appendix: proof of claim 2}.
\end{IEEEproof}

With the above claims, we prove
$\bm{\alpha}^{(1)}\bm{\Delta}_1^{\rm T} <0 $ \af:
\begingroup
\setlength{\jot}{2pt}   
\begin{subequations}
\begin{align}
& \bm{\alpha}^{(1)}\bm{\Delta}^{\rm T}
 \overset{(\ref{eq:alpha^(1), gen PT})}{=}   \sum\nolimits_{k=1}^{t+1}(k-1)\Deltam(k)\\
&\quad  =\sum\nolimits_{k=1}^{r+1}(k-1)\Deltam(k) +\sum\nolimits_{k=r+2}^{t+1}(k-1)\Deltam(k)\\
&\quad  \le \sum\nolimits_{k=1}^{r+1}r\Deltam(k) +\sum\nolimits_{k=r+2}^{t+1}(r+1)\Deltam(k) 
\label{step 1: proof of pos. numerator, thm 3}\\
&\quad < \sum\nolimits_{k=1}^{r+1}(r+1)\Deltam(k)+\sum\nolimits_{k=r+2}^{t+1}(r+1)
\Deltam(k) 
\label{step 2: proof of pos. numerator, thm 3}\\
&\quad =(r+1)\sum\nolimits_{k=1}^{t+1}\Deltam(k)=0,
\label{step 3: proof of pos. numerator, thm 3}
\end{align}
\end{subequations}
\endgroup
where (\ref{step 1: proof of pos. numerator, thm 3}) and (\ref{step 2: proof of pos. numerator, thm 3}) follow from
Claim~\ref{claim:1}, and
(\ref{step 3: proof of pos. numerator, thm 3}) follows from Claim~\ref{claim:2}.
Intuitively, \bcuz $ (\Deltam(k))_{k=1}^{t+1}$ has zero sum and  opposite signs in its first and second halves, weighting $\Deltam$ by a nonnegative, increasing sequence $ \alphabm^{(1)}$ assigns larger weights to its negative entries, resulting in
$\alphabm^{(1)}\Deltam^{\rm T
}<0$.

Next, we prove $\bm{\alpha}^{(2)}\bm{\Delta}^{\rm T}>0$. 
Define 
\be
\label{eq:delta_tq}
\delta_{t,q}(k)\eqdef
\begin{cases}
\Deltam(k) +
\Deltam(t+2-k), &  \mrm{if}\; k\in [r]
\vspace{-.1cm}
\\
\Deltam(r+1),  & \mrm{if} \; k=r+1
\end{cases} 
\ee
\ie,  a pairwise aggregation of the entries of $\Deltam$.
Utilizing the symmetry of  $\bm{\alpha}^{(2)}$ (see (\ref{eq:alpha^(2), gen PT})), \ie,  $\alphabm^{(2)}(k)= \alphabm^{(2)}(t+2-k),\forall k\in[r+1]$, 
we will also regroup the terms in the summation
$\bm{\alpha}^{(2)}\bm{\Delta}^{\rm T}=
\sum_{k=1}^{t+1}\alphabm^{(2)}(k) \Deltam(k)
$ (see (\ref{step 11: alpha(2)*delta>0})). 
Claim~\ref{claim:3} identifies key properties of the sequence $\left\{ \delta_{t,q}(k)\right\}_{k=1}^{r+1} $, which are essential to proving Lemma~\ref{lemma:valid gamma,gen pt}.

\begin{claim}
\label{claim:3}
\tit{The sequence $\{\delta_{t,q}(k)\}_{k=1}^{r+1}$ has the following properties:
\begin{itemize}
    \item[\it 1)] $\delta_{t,q}(1)<0,\delta_{t,q}(r+1)>0$ for any $q,t$ satisfying $q\ge t$;
    \item[\it 2)]  $\exists k'\in[2:r]$ such that  $ \delta_{t,q}(k)<0,\forall k\le k'$;  
    \item[\it 3)] $\forall k\in[2:r-1]$, if $\delta_{t,q}(k+1)<0$, then $\delta_{t,q}(k)<0$;
    \item[\it 4)] Consequently, there exists $k''\in[2:r]$ such that $ \delta_{t,q}(k)<0, \forall k\in [1:k'']$ and $ \delta_{t,q}(k)\ge  0, \forall k\in[k''+1:r+1]$.
\end{itemize}
}\end{claim}
\begin{IEEEproof}
See Appendix~\ref{appendix: proof of claim 3}.
\end{IEEEproof}

Equipped with Claim \ref{claim:3}, we are ready  to  prove $\bm{\alpha}^{(2)}\bm{\Delta}^{\rm T}>0$.
Denote
$  
k_{\rm max} \eqdef \mrm{argmax}_{k\in [2:r]}\{k: \delta_{t,q}(k)<0\} 
$
as the largest index  for which $\delta_{t,q}(k)<0$.
By Claim \ref{claim:3},  $\delta_{t,q}(k)<0,\forall k\in[k_{\rm max}]$ and $\delta_{t,q}(k)\ge 0,\forall k\in[k_{\rm max}+1:r+1]$. 
Define
\be  
\widehat{\alpha}(k)\eqdef \alphabm^{(2)}(k)- k_{\rm max} ,\; k\in[t+1]
\ee 
which is a shifted version of $\alphabm^{(2)}$ (see (\ref{eq:alpha^(2), gen PT})). By construction, 
$ \alphahat(\kmax+1)=0$, 
$\alphahat(k)<0,\forall k\in [\kmax]$, and
$\alphahat(k)>0,\forall k\in [\kmax+2:r+1]$.
Moreover, since $\alphabm^{(2)}$ is \sym,  the shifted version  $\alphahat$ preserves  the same  symmetry, i.e.,
\be
\label{eq:alphahat symmetry}
\alphahat(k)   = \alphahat(t+2-k),\forall k\in[r]
\ee 
We have 
\begin{subequations}
\begin{align}
 \alphabm^{(2)}\Deltam^{\rm T} & = 
\sum_{k=1}^{t+1}\alphabm^{(2)}(k)\Deltam(k) - \kmax \left(\sum_{k=1}^{t+1}\Deltam(k) \right)
\label{step 0: alpha(2)*delta>0, thm 3}\\
&  = 
\sum\nolimits_{k=1}^{t+1} \big(\alphabm^{(2)}(k)-\kmax\big)\Deltam(k)\\
&  = 
\sum\nolimits_{k=1}^{t+1} \alphahat(k)\Deltam(k)\\
& =
\sum\nolimits_{k=1}^r \widehat{\alpha}(k)\Deltam(k) + 
\widehat{\alpha}(r+1)\Deltam(r+1) \notag\\
& 
\quad  +  \sum\nolimits_{k=r+2}^{t+1} \widehat{\alpha}(k)\Deltam(k)\\
& 
= \sum\nolimits_{k=1}^r \widehat{\alpha}(k)\Deltam(k) + 
\widehat{\alpha}(r+1)\Deltam(r+1) \notag \\
& \quad + \sum\nolimits_{i=1}^{r} \widehat{\alpha}(t+2-i)\Deltam(t+2-i)
\label{step 1: alpha(2)*delta>0}
\\
& \overset{(\ref{eq:alphahat symmetry})}{=} 
\sum\nolimits_{k=1}^r \widehat{\alpha}(k)\underbrace{\left(\Deltam(k) + \Deltam(t+2-k)\right)}_{ \overset{(\ref{eq:delta_tq})}{=}  \delta_{t,q}(k)} \notag \\
& \quad\quad+
\widehat{\alpha}(r+1)\Deltam(r+1)
\label{step 11: alpha(2)*delta>0}
\\
& =
\sum\nolimits_{k=1}^{k_{\rm max}} \underbrace{\widehat{\alpha}(k)\delta_{t,q}(k)}_{>0}
+ 
\sum_{k=k_{\rm max}+1}^{r} \underbrace{\widehat{\alpha}(i)\delta_{t,q}(i)}_{\ge 0} \notag\\
& \quad + 
 \underbrace{\widehat{\alpha}(r+1)\delta_{t,q}(r+1)}_{>0}>0,
 \label{step 2: alpha(2)*delta>0}
\end{align}
\end{subequations}
where (\ref{step 0: alpha(2)*delta>0, thm 3}) is \bcuz $\sum_{k=1}^{t+1}\Deltam(k)=0$ \accrdto  Claim \ref{claim:2}.
In (\ref{step 1: alpha(2)*delta>0}), we changed the iteration variables $i=t+2-k$.
(\ref{step 2: alpha(2)*delta>0}) is \bcuz  $\alphahat(k)<0, \delta_{t,q}(k)<0,\forall  k
\in [\kmax]$ \soth $\sum\nolimits_{k=1}^{k_{\rm max}} \widehat{\alpha}(k)\delta_{t,q}(k)>0$;
$\alphahat(k)\ge 0, \delta_{t,q}(k)\ge 0,\forall  k
\in [\kmax+1:r]$ \soth $\sum\nolimits_{k=k_{\rm max}+1}^{r} \widehat{\alpha}(k)\delta_{t,q}(k)\ge 0$, and $\widehat{\alpha}(r+1)\delta_{t,q}(r+1)=\frac{(r-k_{\rm max})}{q}\binom{q}{r-1}\binom{q}{r}>0$.
Hence, $  \alphabm^{(2)} \Deltam^{\rm T}>0$.

 As a result, we have proved $\bm{\alpha}^{(2)}\Deltam ^{\rm T}>0$ and $\bm{\alpha}^{(1)}\Deltam^{\rm T}<0$, yielding $>0$, 
 which completes the proof of Lemma~\ref{lemma:valid gamma,gen pt}.

\section{Discussion}
\label{sec: discussion}

In this section, we discuss further implications of the main theorem, including a \sbp-minimization property of the proposed design and challenges for odd $t$.


\subsection{Minimum \Sbp}
\label{subsec:min sbp over a class,discussion}

One interesting observation is that the proposed \PTd  actually \achvs the minimum \sbp among a set of two-\grp user \grps, as shown  in Lemma \ref{lemma:min sbp over a class,discussion}.
\begin{lemma}
\label{lemma:min sbp over a class,discussion}
\tit{For $(K,t)=(2q+1,2r)$ with the PT design of Theorem~\ref{thm:thm}, the user grouping $\qv=(q+1,q)$ achieves the minimum \sbp  among the set of unequal groupings 
$\left\{ (q_1,K-q_1): q_1\in [q+1:K-t-1]   \right\}$}.
\end{lemma}

\begin{IEEEproof}
Consider any \grpg $\qv=(q_1, K-q_1)$.
We assume $K-q_1\ge t+1$,  \ie, $q_1 \le K-t-1$ \soth  the \sbf and \mgrp types are the  same as (\ref{eq:vk,sk def,gen pt}). Using the \tx \selecs as (\ref{eq:tx select G1, gen PT}) and \ref{eq:tx select G2, gen PT}), the same \gfsv 
$\alphaglobal 
= \lef(0,2,\cdots,t-2, \underline{t}_{r+1} \rig)$ is obtained.
The number of  \tp-$\vv_k$ \sbfs  is equal to $F_{q_1}(\vv_k) = \binom{q_1}{k-1}\binom{K-q_1}{t-k+1},k\in[t+1]$.
We next show $F_{q_1}(\vv_k) $  is strictly \decrsg in $q_1$ for any $k  \in [r]$. \Specly,
\be 
\frac{F_{q_1}(\vv_k)}{F_{q_1+1}(\vv_k)}= 
\frac{q_1 -k+2}{q_1+1}\frac{K-q_1}{K-q_1-t+k-1},k \in[r]
\ee 
Since $t=2r$ and $k\le r$, it holds that
$t-k+1 \ge r+1$, and thus 
$K-q_1-t+k-1 \le K-q_1-r-1$. \Movr, $q_1-k+2\ge  q_1-r+2 $.
\Thf,
\be 
 \label{eq:Fq1(vk)/Fq1+1(vk) bd,disc}
\frac{F_{q_1}(\vv_k)}{F_{q_1+1}(\vv_k)}\ge 
\frac{q_1-r+2}{q_1+1}
\frac{K-q_1}{K-q_1-r-1}.
\ee 
Let $a \eqdef q_1, b\eqdef K-q_1$ \soth  $a+b=K$. The RHS of (\ref{eq:Fq1(vk)/Fq1+1(vk) bd,disc}) exceeds  $1$ if  $b(a-r+2)>(a+1)(b-r-1)$.
Expanding and simplifying yields
$(a + 1)(r + 1) > (r- 1)b$. For $q_1 \in[q+1:K-t-1]$,  we have $a\ge b\ge t-1$. Hence, 
$ (a + 1)(r + 1)-(r- 1)b \ge (b + 1)(r + 1)-(r- 1)b=2b+ r+1\ge 5r -1>0$, which implies  $ {F_{q_1}(\vv_k)}/{F_{q_1+1}(\vv_k)} >1$. \Thf, $ F_{q_1}(\vv_k)> F_{q_1+1}(\vv_k) , \forall k\in[r] $.

Let $ \Fm_{q_1} \eqdef [F_{q_1}(\vv_k)]_{k=1}^{t+1}  $ with $F_{q_1}$, and let $F_{q_1}  \eqdef  \alphaglobal  \Fm_{q_1}^{\rm T}   $ denote the \sbp level under \grpg $(q_1, K-q_1)$.
Define $g(x)\eqdef t\binom{K}{t}-F_x,x \in \{q_1,q_1+1\}$.  
We prove $F_{q_1} \le F_{q_1+1} $ by showing $g(q_1) \ge g(q_1+1)$.
\Ip, by (\ref{eq:Fjcm,thm implication}) and (\ref{eq:alpha^global, gen PT}), we have
$g(q_1) = \sum_{k=1}^{r} (t-\alphaglobal(k)  )F_{q_1}(\vv_k)$, $
g(q_1+1)  = \sum_{k=1}^{r} (t-\alphaglobal(k))F_{q_1+1}(\vv_k)$.
Since $t- \alphaglobal(k)> 0,\forall k\in [r] $, $F_{q_1}(\vv_k)> F_{q_1+1}(\vv_k)    $ directly implies $g(q_1)>g(q_1+1)$.
Hence, we have proved $F_{q_1} < F_{q_1+1}, \forall q_1$. \Aar, $q_1= q+1 $ minimizes the \sbp with the \grpg $(q+1,q)$, completing the proof of Lemma  \ref{lemma:min sbp over a class,discussion}.
\end{IEEEproof}

\subsection{The Case of Odd $t$}
\label{subsec:odd t,disc}

The PT design in \Thm \ref{thm:thm} applies only to even $t$. In this section, we examine the case of odd $t$, review existing constructions, and identify the structural obstacles that arise under \het  \sbp.

\tit{Existing PT Designs:}
In one prior work~\cite{zhang2026taming}, we proposed \homo-\sbp  \PTd for $(K,t)=(2q+1,2q-1)$  that \achvs an order-wise \red compared with JCM. However, this construction  applies only to the specific case $t=K-2$, which lacks generality.

\tit{Challenge of Odd  $t$ under \Het Design:}
We briefly explain why the even-$t$ \PTd in \Sec \ref{subsec: general pt design} does not extend to odd  $t$. 
The core difficulty lies in the \tit{misalignment of local FS factors} across \mgrp \tps near the pivot.
\Specly, (\ref{eq:tx select G2, gen PT}) induces two pivot types $\sv_{r+1}=( r^\dagger,r+1 )  $ and $\sv_{r+2}=(r+1, r^\dagger)$,
where the transmitters switch between components as the local FS contributions reverse.  
For the involved subfile types $ \Ic_{r+1}=\{\vv_{r},\vv_{r+1} \}$ and $ \Ic_{r+2}=\{\vv_{r+1},\vv_{r+2} \}$, 
the local FS vectors coincide at type $\vv_{r+1}$ (see Table \ref{tab:fsfs at pivot types,disc}(a)), \ie, 
$ \alphabm_{r+1}=( \underline{\star}_{r-1}, r-1,r, \underline{\star}_r) $, 
$ \alphabm_{r+2}=(\underline{\star}_{r} ,r,r-1,\underline{\star}_{r-1}) $, 
thus yielding the ``hill-shaped" \gfsv  $\alphabm^{(2)}$ in (\ref{eq:alpha^(2), gen PT}).
Combined with $ \alphabm^{(1)}$ in (\ref{eq:alpha^(1), gen PT}), the   \agg  \gfsv  $\alphaglobal$ in (\ref{eq:alpha^global, gen PT}) \incrss  linearly   with unit slope  from zero and is capped at $t$ (see Fig. \ref{fig:alphas comparison,gen pt}).
\begin{table}[t]
\centering
\caption{\small FS Factors at Pivot Types}
\label{tab:fsfs at pivot types,disc}
\vspace{-.1cm}
\setlength{\tabcolsep}{1pt}
\renewcommand{\arraystretch}{0.8}
\begin{subtable}[t]{0.48\columnwidth}
\centering
\caption{Even $t=2r$}
\vspace{-.1cm}
\begin{tabular}{|c|c|c|c|}
\hline
 & $\vv_{r}$ & $\vv_{r+1}$ & $\vv_{r+2}$ \\
\thickhline
$\sv_{r+1}$ & $r-1 $&  $r$  &  $\star$ \\
\hline
$\sv_{r+2}$ & $\star$ & $r $& $r-1$  \\
\thickhline
$\alphabm^{(2)}$ & $r-1$ & $r$ & $r-1$ \\
\hline
\end{tabular}
\end{subtable}
\hfill
\begin{subtable}[t]{0.48\columnwidth}
\centering
\caption{Odd $t=2r+1$}
\vspace{-.1cm}
\begin{tabular}{|c|c|c|c|}
\hline
 & $\vv_{r}$ & $\vv_{r+1}$ & $\vv_{r+2}$\\
\thickhline
$\sv_{r+1}$ & $r-1 $&  $r$  &  $\star$   \\
\hline
$\sv_{r+2}$ & $\star$ & $r+1 $& $r$    \\
\thickhline
$\alphabm^{(2)}$ & $ r^2-1 $ & $r(r+1)$ & $r^2$  \\
\hline
\end{tabular}
\end{subtable}
\end{table}

When  $ t=2r+1$, the two pivot \tps $\sv_{r+1}$ and $\sv_{r+2}$ cannot generate perfectly aligned FS factors for subfile type $\vv_r$.
To obtain a hill-shaped
$\alphabm^{(2)}$ similar to the even-$t$ case, the second component  of $\sv_{r+2}$
must be selected as \txs, 
\ie, $\sv_{r+2}=(r+1, (r+1)^\dagger)$.
Under this choice, type $\vv_{r+1}$ acquires FS factors $r$ and $r+1$ from $\sv_{r+1}$ and $\sv_{r+2}$ (see Table \ref{tab:fsfs at pivot types,disc}(b)), \resp.
Since $\gcd(r,r+1)=1$,
the vector LCM operation yields a global FS factor at least $r(r+1)$, which exceeds $ t=2r+1$ for all $r\ge 2$.
Consequently, the aggregate FS vector cannot be upper bounded by $t$, and the structural alignment observed in the even-$t$ case becomes unattainable.
This obstruction is \tit{intrinsic to the parity of $t$}: when $t$ is odd, the two central FS contributions differ by one and are coprime, forcing a multiplicative growth of their LCM.

Nevertheless, \sbp reduction remains possible for small values of $t$, as shown in the \folwg \ex.

\begin{example}[Odd $t=3$]
\label{example:het pt for t=3,disc}
Let $(K,t)=(2q+1,3)$ with $q\ge 4$.  Under  user grouping $\qv=(q+1,q)$, the subfile and \mgrp types are 
$ \vv_1=(0,3), \vv_2=(1,2),\vv_3=(2,1),\vv_4=(3,0)$, and $\sv_1=(0,4), \sv_2=(1,3),\sv_3=(2,2), \sv_4=(3,1),\sv_5=(4,0)$ with \invod \sbf types
$\Ic_1 =\{\vv_1\}, \Ic_2 =\{\vv_1,\vv_2\}, 
\Ic_3 =\{\vv_2,\vv_3\}, \Ic_4 =\{\vv_3,\vv_4\}$, and $\Ic_5 =\{\vv_4\}$.
The \tx \selecs are
\begin{subequations}
\begin{align}
 \Gc_1:\;\sv_1 & =\left(0,4^{\dagger}\right),
\sv_2=\left(1^{\dagger},3\right),
\sv_3=\left(2^{\dagger},2\right),\notag \\
 \sv_4 & =\left(3^{\dagger},1\right),
\sv_5=\left(4^{\dagger},0\right),\\
\Gc_2:\; \sv_1 & =\left(0,4^{\dagger}\right),
\sv_2=\left(1^{\dagger},3\right),
\sv_3=\left(2,2^{\dagger}\right),\notag \\
\sv_4 & =\left(3,1^{\dagger}\right),
\sv_5=\left(4^{\dagger},0\right),
\end{align}
\end{subequations}
yielding 
$\bm{\alpha}^{(1)}=(0,1,2,3)$, 
$\bm{\alpha}^{(2)}=(0,2,1,0)$, and $\alphaglobal=(0,3,3,0)$.
According to (\ref{eq:gamma, gen pt}), the packet  size ratio is equal to
$\gamma = 
-\frac{ \bm{\alpha}^{(1)}\bm{\Delta}^{\rm T} }{\bm{\alpha}^{(2)}\bm{\Delta}^{\rm T}}
= \frac{2(2q-1)}{q+4}>0$, 
where $\bm{\Delta} =(q-1)\big[ \frac{ (q-2)}{2},\frac{ (q+2)}{2},-\frac{q }{2},-\frac{q }{2}\big]$. 
Thus, 
$\Fpt=\alphaglobal\Fm^{\rm T}=\frac{3q(2q-1)(q+1)}{2}$, and
$
\frac{F_{\rm PT}}{F_{\rm JCM} } = 
\frac{3}{4}\left(1+\frac{1}{K} \right)\le 5/6,\forall K\ge 9
$ and $\lim_{K\to\infty} \frac{F_{\rm PT}}{F_{\rm JCM} }=3/4$, implying a one-fourth packet reduction asymptotically. 
\hfill $\lozenge$
\end{example}

\section{Conclusion}
\label{sec:conclusion}
In this paper, we introduced a heterogeneous \sbp extension of the packet type (PT)-based framework for \dd coded caching. By allowing packet sizes to vary across refined packet types, the proposed design overcomes an inherent limitation of existing homogeneous PT constructions. Through  the joint optimization of  user grouping, multicast transmitter selection,
and packet sizing,
we established a new class of rate-optimal schemes for $(K,t)=(2q+1,2r)$ that achieves a provable constant-factor reduction in \sbp over  \sota \schms.
The proposed heterogeneous \ptf reveals that unequal user grouping, combined with coordinated transmitter selection across \cpgrps,  provides
additional flexibility in satisfying memory constraints that cannot be attained under homogeneous splitting. 

Several directions deserve further investigation, including:
i) exploring more general user \grpgs, more diverse \tx \selec strategies, and \pkt sizing
for broader  $(K,t)$ regimes; and ii) relaxing the \ptf to allow suboptimal \comm rates
in exchange for further \sbp reduction. Ultimately, our goal is to characterize the fundamental trade-off between communication rate and \sbp.

\appendices

\section{Proof of Lemma \ref{lemma:monotonicity of Fpt/Fjcm}}
\label{sec:proof of ratio monocity,app} 
Define pmf 
$ 
p_k(q) \eqdef {f_k}/{\sum_{k=1}^{t+1}f_k}, k \in [t+1]$.
Let $J_q \sim \mrm{Hypergeo}(2q+1,q+1,t)$ be a \hygeo \randvar with pmf
\be
\label{eq:hypergeo pmf,app}
\mrm{Pr}(J_q=j)= \frac{\binom{q+1}{j}\binom{q}{t-j} }{\binom{2q+1}{t}},\;j \in  [0:t]
\ee
Since $p_k(q)= \mrm{Pr}(J_q=k-1) $, we can write 
\be
\label{eq:rho,first appear,app}
\rho_t(q)=\sum_{k=1}^{t+1}\frac{\alphaglobal(k)}{t}p_k(q)= \mbb{E}\lef[h(J_q)  \rig],
\ee
where
\be 
h(j)=
\begin{cases}
2j/t, & \mrm{if}\;j \in [0:r-1]  \vspace{-.15cm} \\
t,  & \mrm{if}\;j \in [r:t]
\end{cases}
\ee
which is nondecreasing and strictly increasing on $[0:r-1]$.
For fixed $t$, the \hygeo family is \stochaly  decreasing in $q$ \bcuz the population odds $ \omega \eqdef 1+1/q$ strictly \decrs with  $q$~\cite{shaked1994stochastic}.
Hence,  $J_{q+1}$ is dominated by  $J_q$, 
\be
J_{q+1}  \prec_{\rm st} J_q
\ee 
and the dominance is strict since $\mbb{E}[J_q]=\frac{t(q+1)}{2q+1}$ is strictly decreasing in $ q$. 
Because $h(\cdot)$ is nondecreasing and not constant, strict stochastic dominance yields
$\mbb{E}\lef[h(J_{q+1 })\rig]< \mbb{E}\lef[h(J_q)  \rig]$, \ie, 
$\rho_t(q+1 )< \rho_t(q)$.
This completes the proof of Lemma \ref{lemma:monotonicity of Fpt/Fjcm}.

\section{Proof of Cache \MemoConst}
\label{sec:proof of Zu MC, appendix}

Denote $ \alphabm_\gamma \eqdef \alphabm^{(1)} +  \gamma \alphabm^{(2)} $.
For any \user{u\in \Qc_i,i=1,2}, 
\begin{subequations}
\begin{align}
& |Z_u|/N \notag\\
& =\ell^{(1)} \sum_{k=1}^{t+1}\alphabm^{(1)}(k)F_i(\vv_k)
+ \ell^{(2)} \sum_{k=1}^{t+1}\alphabm^{(2)}(k)F_i(\vv_k)\\
& = \ell^{(1)} \sum_{k=1}^{t+1}\lef(\alphabm^{(1)}(k)+\gamma\alphabm^{(2)}(k)\rig)F_i(\vv_k)\\
& = \ell^{(1)}  \alphabm_\gamma  \Fm_i^{\rm T} \\
&  \overset{(\ref{eq:pkt size, gen pt})}{=}
 L(\alphabm_\gamma  \Fm_i^{\rm T})/(\alphabm_\gamma  \Fm^{\rm T}),
\end{align}
\end{subequations}
where $\Fm, \Fm_1$ and $\Fm_2$ are given in (\ref{eq:F,gen pt}), (\ref{eq:F1,gen pt}) and (\ref{eq:F2,gen pt}), \resp. 

Now we show $(\alphabm_\gamma  \Fm_i^{\rm T})/(\alphabm_\gamma  \Fm^{\rm T})=t/K,\forall i\in[2]$. By the binomial identity $\binom{m+1}{n+1}=\frac{m+1}{n+1}\binom{m}{n}$,  we have
\begin{align}
\label{eq:Fi relation w/ F, app}
F_1(\vv_k) = \frac{k-1}{q+1}F(\vv_k), \;
F_2(\vv_k) = \frac{t-k+1}{q}F(\vv_k).
\end{align}
Hence, 
\be 
\label{eq:Deltam(k), app}
\Deltam(k) =F_2(\vv_k)- F_1(\vv_k)=F(\vv_k)\lef(  \frac{t-k+1}{q}-  \frac{k-1}{q+1} \rig).
\ee 
The \memoconst (\ref{eq:differential MC, gen pt}) requires $ \alphabm_\gamma \Deltam^{\rm T} =0$, \ie, 
\be 
\sum_{k=1}^{t+1} \alphabm_\gamma(k)\Deltam(k)=0.
\ee 
Substitute (\ref{eq:Deltam(k), app}) and multiply by  $q(q+1)$, we get
\begin{align}
& \sum_{k=1}^{t+1} \alphabm_\gamma(k)F(\vv_k)\lef(  (q+1)(t-k+1)-q(k-1) \rig) \notag\\
&=\sum_{k=1}^{t+1} \alphabm_\gamma(k)F(\vv_k)\lef(  (q+1)t-(2q+1)(k-1) \rig)=0,
\end{align}
which implies
\be
\label{eq: crucial id, app}
(q+1)t\sum_{k=1}^{t+1} \alphabm_\gamma(k)F(\vv_k)
=
(2q+1)\sum_{k=1}^{t+1} \alphabm_\gamma(k)(k-1)F(\vv_k).
\ee 
Next compute the desired ratio using (\ref{eq:Fi relation w/ F, app}):
\begin{subequations}
\begin{align}
\frac{\alphabm_\gamma \Fm_1^{\rm T} }{\alphabm_\gamma \Fm^{\rm T}} 
& =
\frac{\sum_{k=1}^{t+1}\alphabm_\gamma(k)F_1(\vv_k)}{\sum_{k=1}^{t+1}\alphabm_\gamma(k)F(\vv_k)}\\
& \overset{(\ref{eq:Fi relation w/ F, app})}{=}
\frac{\sum_{k=1}^{t+1}\alphabm_\gamma(k)\frac{k-1}{q+1}F(\vv_k)}{\sum_{k=1}^{t+1}\alphabm_\gamma(k)F(\vv_k)}\\
&= \frac{1}{q+1}\cdot \frac{\sum_{k=1}^{t+1}\alphabm_\gamma(k)(k-1)F(\vv_k)}{\sum_{k=1}^{t+1}\alphabm_\gamma(k)F(\vv_k)}\\
& \overset{(\ref{eq: crucial id, app})}{=}
\frac{1}{q+1}\cdot \frac{\frac{(q+1)t}{2q+1} \sum_{k=1}^{t+1} \alphabm_\gamma(k)F(\vv_k)}{\sum_{k=1}^{t+1}\alphabm_\gamma(k)F(\vv_k)}\\
& =\frac{t}{2q+1}=t/K.
\end{align}
\end{subequations}
Since $\alphabm_\gamma\Deltam^{\rm T}=\alphabm_\gamma(\Fm_2^{\rm T}-\Fm_1^{\rm T}   )=0$, $\alphabm_\gamma\Fm_1^{\rm T}=\alphabm_\gamma\Fm_2^{\rm T}$. 
\Aar, $|Z_u|=NL(t/K)=ML, \forall u \in [K]$, proving the satisfaction of the \memoconst.

\section{Proof of Claims}
\label{appendix: proof of calims}

\subsection{Proof of Claim~\ref{claim:1}}
\label{appendix: proof of claim 1}
We have 
$\Deltam(1) = \binom{q-1}{t-1}$, $\Deltam(t+1)=-\binom{q}{t-1} $, and for any  $k\in [2:t+1]$,
\begin{subequations}
\begin{align}
& \Deltam(k)   =F_2(\vv_k)-F_1(\vv_k)\\
&\quad 
=\binom{q+1}{k-1}\binom{q-1}{t-k}- \binom{q}{k-2}\binom{q}{t-k+1}\\
& \quad  
=  \underbrace{ \binom{q}{k-2}\binom{q}{t-k+1}}_{\eqdef f_{t,q}(k)}
\bigg(\underbrace{\frac{q+1}{q}\frac{t-k+1}{k-1} -1 }
_{\eqdef \varepsilon_{t,q}(k)}\bigg),
\label{eq:step0,proof of claim 1}
\end{align}
\end{subequations}
where in (\ref{eq:step0,proof of claim 1})  we used the combinatorial identity $\binom{n+1}{m+1}=\frac{n+1}{m+1}\binom{n}{m}$.
It is easy to see  that $\varepsilon_{t,q}(k)$ is strictly
decreasing with $i$ for fixed $t$ and $q\ge t+1$.
Since $\varepsilon_{t,q}(r+1)=1/q>0$ and
$ 
\varepsilon_{t,q}(r+2) = \frac{ q+1}{q}\frac{r-1}{r+1}- 1
\le \frac{ t+1}{t}\frac{r-1}{r+1}- 1\le \frac{ t+2}{t}\frac{r-1}{r+1}- 1= -1/r
$, we have $\varepsilon_{t,q}(k)>0,\forall k\in [2: r+1]$ 
and $\varepsilon_{t,q}(k)<0,\forall k\in [r+2:t+1]$. Because $f_{t,q}(k)>0,\forall k$, we conclude that
$\Deltam(k)>0$, $\forall k\in [r+1]$ and
$\Deltam(k)<0$,$\forall k\in[r+2:t+1]$.

\subsection{Proof of Claim~\ref{claim:2}}
\label{appendix: proof of claim 2}
Let $|\Qc_1|=q+1$ and $|\Qc_2|=q$. 
The number of \tp-$\vv_k=(k-1,t-k+1),k \in[t+1]$ \sbfs stored by each user in $\Qc_1$ and $\Qc_2$ is equal to 
$F_1(\vv_k)  =\binom{|\mc{Q}_1|-1}{k-2}\binom{|\mc{Q}_2|}{t-k+1}$ and $F_2(\vv_k) =\binom{|\mc{Q}_1|}{k-1}\binom{|\mc{Q}_2|}{t-k}$, \resp.
Hence, the total number of subfiles stored by each user in $\Qc_1$ equals 
\begin{subequations}
\label{eq:Q1,proof of claim 2}
\begin{align}
\sum_{k=1}^{t+1}F_1(\vv_k)&=\sum_{k=2}^{t+1}\binom{|\mc{Q}_1|-1}{k-2}\binom{|\mc{Q}_2|}{t-k+1}\\
&=\sum_{k=0}^{t-1}\binom{|\mc{Q}_1|-1}{k}\binom{|\mc{Q}_2|}{t-k-1}       \\
&=\binom{|\mc{Q}_1|+|\mc{Q}_2|-1}{t-1}
= 
\binom{K-1}{t-1},
\label{step: total no. stored subfiles Q1, thm 3}
\end{align}
\end{subequations}
where Vandermonde's identity is applied in (\ref{step: total no. stored subfiles Q1, thm 3}). Similarly, the total number of subfiles stored by each user in $\Qc_2$ equals
\begin{subequations}
\label{eq:Q2,proof of claim 2}
\begin{align}
\sum_{k=1}^{t+1}F_2(\vv_k) & =\sum_{k=1}^{t}\binom{|\mc{Q}_1|}{k-1}\binom{|\mc{Q}_2|-1}{t-k}\\
& =\sum_{k=0}^{t-1}\binom{|\mc{Q}_1|}{k}\binom{|\mc{Q}_2|-1}{t-k-1}\\
& = \binom{|\mc{Q}_1|+|\mc{Q}_2|-1}{t-1}
= 
\binom{K-1}{t-1}.
\end{align}
\end{subequations}

\Aar, we have 
$
\sum_{k=1}^{t+1}\Deltam(k) =
\sum_{k=1}^{t+1}(F_2(\vv_k)- F_1(\vv_k)) 
= \sum_{k=1}^{t+1}F_2(\vv_k) - \sum_{k=1}^{t+1}F_1(\vv_k) \overset{(\ref{eq:Q1,proof of claim 2}),(\ref{eq:Q2,proof of claim 2})}{=}  0  
$, which completes the proof of Claim~\ref{claim:2}.

\subsection{Proof of Claim~\ref{claim:3}}
\label{appendix: proof of claim 3}
Assume $r\ge 2$. 
Recall that $\Deltam(1) = \binom{q-1}{t-1}$, $\Deltam(t+1)=-\binom{q}{t-1} $, and 
$\Deltam(k)=f_{t,q}(k) \varepsilon_{t,q}(k)$, where
$f_{t,q}(k)\eqdef \binom{q}{k-2}\binom{q}{t-k+1}$, $\varepsilon_{t,q}(k) \eqdef \frac{q+1}{q}\frac{t-k+1}{k-1} -1$, $k\in [2:t+1]$ as defined  in (\ref{eq:step0,proof of claim 1}). We prove  the properties \resp.

\tit{Property 1):} It is  easy to see that $\delta_{t,q}(1)= \Deltam(1) +\Deltam(t+1)=\binom{q-1}{t-1}-  \binom{q}{t-1} = -\binom{q-1}{t-2}<0$, and 
$\delta_{t,q}(r+1) = \Deltam(r+1)= \varepsilon_{t,q}(r+1) f_{t,q}(r+1) =  f_{t,q}(r+1)/q>0$.

\tit{Property 2):}
For any  $k\in [2:r]$, 
\begin{subequations}
\begin{align}
& \delta_{t,q}(k)  =\Deltam(k) + \Deltam(t+2-k) \\ 
& =\underbrace{\varepsilon_{t,q}(k)}_{>0}f_{t,q}(k) +\underbrace{ \varepsilon_{t,q}(t+2-k)}_{<0}f_{t,q}(t+2-k)
\label{step0,property 2} 
\\
&=\varepsilon_{t,q}(k)f_{t,q}(k) + \varepsilon_{t,q}(t+2-k)f_{t,q}(k+1), \label{step1,property 2}  
\end{align} 
\end{subequations}
where (\ref{step0,property 2}) is due to Claim \ref{claim:1}, and (\ref{step1,property 2}) is \bcuz 
$f_{t,q}(k+1)   = f_{t,q}(t+2-k),\forall k\in  [2:t+1]$ (see (\ref{eq:step0,proof of claim 1})).
Define 
\be 
\label{eq:def V(k),proof claim 3}
V(k) \eqdef  \frac{\varepsilon_{t,q}(t+2-k)}{-\varepsilon_{t,q}(k)}\cdot  \frac{f_{t,q}(k+1)}{f_{t,q}(k)},  \;  k\in [2:r] 
\ee 
Then $V(k)>1 \Leftrightarrow  \delta_{t,q}(k)<0$. Moreover,
\be 
\label{eq:full V(k),proof claim 3}
V(k) = \frac{\lef( q\left(t-2(k-1)  \right)-(k-1)  \rig)(q-k+2)}{\lef( q\left(t-2(k-1)  \right)+t-k+1  \rig)   (q-t+k)} 
\eqdef \frac{x(k) }{y(k)}
\ee 
We prove that there exists an index $k'\in[2:r]$ \suth $V(k)>1, \forall k \in [2:k']$.
Since $q \ge t$,  both the numerator and denominator  of $V(k)$ are positive, \ie, $x(k),y(k) >0, \forall k \in [2:r]$. Thus,  
$ x(k)-y(k) >0 \Leftrightarrow V(k)>1$.
A direct calculation gives
\begin{align}
\label{eq:x(k)-y(k),proof claim 3}
x(k)-y(k) & =q\lef((t-2(k-1) )^2 -t\rig)  +A,
\end{align}
where $A \eqdef(k-1)(k-2)+(t-k)(t-k+1)  >0  $ since for any $k \in [2:r]$, $(k-1)(k-2)\ge 0, (t-k)(t-k+1)\ge r(r+1)>0$.
Thus, if $g(k) \eqdef (t-2(k-1) )^2 -t \ge 0  $, then (\ref{eq:x(k)-y(k),proof claim 3}) yields $x(k)-y(k)>0$, and consequently  $V(k) >1$. 

Over $k  \in [2:r]$, the term $t-2(k-1)$ is positive and strictly
decreases with $k$, hence $g(k)$ is strictly decreasing. \Inadd, assuming $t\ge 4 (r\ge 2)$, then 
$g(2)=(t-2)^2-t\ge 0$, $g(r)=(t-2(r-1))^2-t=4-t \le 0$. 
\Thf, there exists $k'\in [2:r]$ \suth $g(k)\ge 0,\forall k\le k'$, which implies $V(k)>1,\forall k\le k'$. Using the established equivalence $V(k)>1 \Leftrightarrow \delta_{t,q}(k)<0$, Property 2) follows.

\tit{Property 3):}
Consider the condition that ensures $x(k)-y(k)> 0$ (thus $ \delta_{t,q}(k)<0$) when $k\in [k'+1:r]$ such that $(t-2(k-1))^2<t$. 
This requires
\be
\label{eq:q upper bound}
q < \frac{A}{t-(t-2(k-1))^2    }
\eqdef \varphi_t(k).
\ee 
Claim \ref{claim:4} (see below) shows that
$\varphi_t(k)$ is strictly decreasing in $k$. 
Thus, for any $k \in [2:r-1]$, $q < \varphi_t(k+1)$ ($\Rightarrow \delta_{t,q}(k+1) <0$) implies  $q < \varphi_t(k)$---yielding $ \delta_{t,q}(k) <0$---\bcuz $\varphi_t(k)>\varphi_t(k+1)$. \Thf, Property 3) is proved.

\tit{Property 4):} Properties 1), 2), and 3) together imply Property 4).
Let $ \varphi_t(k) \eqdef A_k/B_k$, where $A_k\eqdef(k-1)(k-2)+(t-k)(t-k+1)$,
$B_k= t-(t-2(k-1))^2,k \in [2:r]$. 
Define the positive-denominator index set as
$\Bc \eqdef \big[ \lfloor \frac{t-\sqrt{t}}{2}  \rfloor+2: \lceil \frac{t+\sqrt{t}}{2} \rceil       \big]$ on which $B_k>0, \forall k\in \Bc$.

\begin{claim}[Monotonicity of $\varphi_t$ on $\Bc$]
\label{claim:4}
\tit{For any $k$ \suth $k,k+1 \in \Bc$ and $k\le r-1$, we have $\varphi_t(k+1)<\varphi_t(k)$.}
\end{claim}

\begin{IEEEproof} 
Consider any $k,k+1\in\Bc$ \soth $B_k,B_{k+1}>0$.
\be
\label{eq:phi(k+1)-phi(k),proof claim 4}
\varphi_t(k+1)-\varphi_t(k)= \frac{A_{k+1}B_k-A_kB_{k+1}}{B_kB_{k+1}}.
\ee 
Since $B_kB_{k+1}>0$, it suffices to determine the sign of  $\beta_k \eqdef A_{k+1}B_k-A_kB_{k+1}$.  A direct expansion gives
\be 
\label{eq:beta,proof claim 4}
\beta_k =  2t(t-1)(2k-(t+1)),
\ee 
which is negative if $k\le r-1< \frac{t+1}{2}=r+1/2$. \Aar, $\varphi_t(k+1)-\varphi_t(k)<0$, completing the proof of Claim \ref{claim:4}. 
\end{IEEEproof}


\bibliographystyle{IEEEtran}
\bibliography{references_newest}
\end{document}